\documentclass[10pt]{article}
\usepackage{graphicx} % Required for inserting images
\usepackage{tabularx} % extra features for tabular environment
\usepackage{amsmath}  % improve math presentation
\usepackage{fancyhdr}
\usepackage{amssymb}
\usepackage{physics}
\usepackage[]{siunitx}
\usepackage{setspace}
\usepackage{dsfont}
\usepackage{tikz}
\usepackage{ bbold }
\usepackage[english]{babel}
\usepackage{amsthm}
\usepackage{mathtools}

\newtheorem{theorem}{Theorem}[section]

\usepackage{cite}
\usepackage{appendix}
\usepackage[final]{hyperref} % adds hyper links inside the generated pdf file
\hypersetup{
	colorlinks=true,       % false: boxed links; true: colored links
	linkcolor=black,        % color of internal links
	citecolor=black,        % color of links to bibliography
	filecolor=magenta,     % color of file links
	urlcolor=black         
}

\title{Indefinite Causal Structure and Causal Inequalities with Time-Symmetry}
\author{
    Luke Mrini and Lucien Hardy \\
    \normalsize \textit{Perimeter Institute for Theoretical Physics,} \\ \
    \normalsize \textit{31 Caroline Street North,} \\ 
    \normalsize \textit{Waterloo, Ontario N2L 2Y5, Canada}
}

\date{}
\begin{document}

\maketitle

\begin{abstract}
    Time\textendash reversal symmetry is a prevalent feature of microscopic physics, including operational quantum theory and classical general relativity. Previous works have studied indefinite causal structure using the language of operational quantum theory, however, these rely on time-asymmetric conditions to constrain both operations and the process matrix. Here, we use time-symmetric, operational probabilistic theory to develop a time-symmetric process matrix formalism for indefinite causal structure. This framework allows for more processes than previously considered and a larger set of causal inequalities. We demonstrate that this larger set of causal inequalities offers new opportunities for device\textendash independent certification of causal non-separability by violating new inequalities. Additionally, we determined that the larger class of time-symmetric processes found here is equivalent to those with Indefinite Causal Order and Time Direction (ICOTD) considered by Chiribella and Liu \cite{liu2024tsirelson}, thereby providing a description of these processes in terms of process matrices.
\end{abstract}

\tableofcontents

\section{Introduction}
\label{intro}

The notion of causal structure---the relations between events that characterize them as timelike, spacelike, or lightlike separated---is of a central importance in information theory and in fundamental physics. The distinction between past and future, or lack thereof, plays a foundational role in quantum theory, classical general relativity, information theory, and physics more broadly. We use the term \emph{time-symmetry} to describe a theory which is invariant under some form of time-reversal, whether it be the $\mathcal{CPT}$-invariance of the Standard Model, or any combination of $\mathcal{T}$ with other discrete transformations. In this work, we develop a formalism to explicate the full range of causal structures that are possible in a quantum theory with time-symmetry, specifically in the setting of two parties acting on finite-dimensional Hilbert spaces. Our work establishes a clear relationship between formalisms for causal structures with and without time-symmetry. We also derive a set of causal inequalities in a theory-independent manner whose violation can characterize certain exotic causal structures as lying outside the realm of classical causality.

The range of circuit architectures that may be used in an information processing setting are limited by the possible causal structures and time-orderings available among circuit elements. Traditional circuits using classical or quantum information rely only on definite timelike separations between circuit elements in order to signal forward in time and perform computations. A wider range of causal structures may be utilized in information processing when quantum mechanical considerations are taken into account. These hypothesized, indefinite causal structures have been shown to permit new types of circuits that offer a computational advantage over traditional quantum circuits \cite{Hardy:2007fv,2021PhRvR...3d3012R,2016arXiv161205099P,DeAraujoSantos:2016nit, Chiribella:2009lvz} and enhanced communication protocols \cite{Ebler:2018zrj}. 
%hypothesized new indefintie causal structure, not quantum causal structure

Another motivation for studying exotic causal structures comes from quantum gravity where it is expected that superpositions of the spacetime metric, which determines the causal relationship between any pair of points, will play an important role \cite{Penrose:1996cv,delaHamette:2021iwx, Giacomini:2021aof}. Quantum indefiniteness of the metric gives rise to \emph{indefinite causal structure} \cite{hardy2005probability}, an area of research that has seen increasing attention in recent years \cite{Zych:2017tau,PhysRevX.8.011047}. Oreshkov, Costa, and Brukner \cite{Oreshkov:2011er} demonstrated that if one only assumes that quantum mechanics holds locally in its usual form, without imposing a global causal structure, it is possible that exotic \emph{causally non-separable processes} may emerge \cite{Oreshkov:2016mbt}. They demonstrated this by studying a class of superoperators called \emph{process matrices} that generalize the notion of a density matrix and assign probabilities to operational circuits. Causally non-separable processes have associated probability distributions that are not consistent with any convex mixture of definite causal structures. Some of these processes violate causal inequalities \cite{Branciard_2016}, demonstrating inconsistency with definite causal structure as Bell's inequalities do for non-locality. 

One example of a causally non-separable process---the quantum switch\cite{Chiribella:2009lvz}---has been realized in the laboratory \cite{Rubino:2017upz, 2015NatCo...6.7913P} and has been shown to violate a set of device-independent causal inequalities \cite{vanderLugt:2022eqz}. These experiments presumably took place in a definite spacetime, therefore in these cases the causal non-separability of the quantum switch is not due to spacetime indefiniteness in any quantum gravity sense. Rather, it is likely due to a phenomenon such as ``time-delocalized quantum systems,'' proposed by Oreshkov \cite{oreshkov2019time}. This is discussed further at the end of Section~\ref{ts process}. It remains to be seen whether other exotic causal structures besides the quantum switch may be realized in physical circuits, or whether causal non-separability may be produced experimentally involving genuine spacetime indefiniteness.

%future choice
It is conventional wisdom that signalling only happens forward in time, that is, a past event must not be influenced by a future choice. Existing work on indefinite causal structure relies on this time asymmetry by constraining local operations to prohibit signalling backwards in time. An \emph{operation}, as explained in Appendix~\ref{meas}, refers to a party acting in a compact region of spacetime who may transform physical systems and report classical information. This is a central concept to the current work and will be expanded on throughout the text. Similarly, the causal inequalities are derived by supposing that processes with definite causal structure satisfy certain time-asymmetric no-signalling constraints \cite{Branciard_2016}. This unequal treatment of past and future lies in tension with the time-reversal symmetry of classical general relativity. In fact, time-symmetry is a prevalent feature of microscopic physics more generally. The Standard Model of Particle Physics is invariant under the simultaneous reversal of charge, parity, and time known as $\mathcal{CPT}$-invariance \cite{greaves2014cpt}. This example demonstrates how time-symmetric microscopic physics does not necessarily entail invariance under time-reversal alone, but that it might be accompanied by other discrete transformations. 

Dynamics in quantum theory as governed by the Schr\"{o}dinger equation are time-symmetric. The standard treatment of the measurement process in quantum theory introduces time asymmetry, however, Aharonov, Bergmann and Lebowitz (ABL) showed in the 1960's that the von Neumann model of measurement can be formulated in a time-symmetric fashion at the microscopic level \cite{1964PhRv..134.1410A}. Standard operational quantum theory is time-asymmetric because operations are constrained to be trace non-increasing in the forward time direction. However, a time-symmetric formulation of operational quantum theory is also possible \cite{Oreshkov:2015aha, DiBiagio2021,Hardy:2021fqs}. A time-symmetric operational probabilistic theory (TSOPT), of which operational quantum theory is a special case, was recently constructed in \cite{Hardy:2021fqs}. At the level of macroscopic statistical physics, time-asymmetry necessarily arises due to the second law of thermodynamics, but this can always be reduced to time-symmetric microscopic physics. Therefore, it is of interest to develop a time-symmetric formalism for indefinite causal structure---this is the primary goal of the current work.

Existing work by Chiribella and Liu \cite{Chiribella:2020yfr} partially addresses this question from a different perspective by studying the case where the direction of the flow of time through an operation is indefinite. In this work, a broader class of processes was uncovered relative to those that can be described in the time-forward process matrix formalism of Ref.~\cite{Oreshkov:2011er}, offering new computational advantages \cite{Liu:2022auq, Guerin:2016qhz}. They also demonstrated recently that processes with both indefinite causal order and time direction (ICOTD) can maximally violate any causal inequality \cite{liu2024tsirelson}. Indefinite causal order refers to the causal ordering between multiple parties, while indefinite time direction refers to the flow of time within a single party's operation. We sometimes use ``indefinite causal structure'' as a generic term to refer to either or both of these. The ICOTD processes coincide with the full set of processes allowed in our time-symmetric process matrix framework. Chiribella and Liu obtain this set of processes by studying so-called ``bidirectional devices,'' devices having the property that an input-output inversion results in another valid operation. Examples of such devices include half-wave and quarter-wave plates in quantum optics. A process with indefinite time direction known as the Quantum Time Flip has been realized experimentally by Guo et al. \cite{Guo:2022jyf}. The approach of Chiribella and Liu does not manifestly incorporate time-symmetry since still the conditions used to constrain physical operations treat the past and future distinctly. To address this concern in the current work, we treat operations in a time-symmetric way from the onset. Our approach offers a unified framework to describe ICOTD process in terms of process matrices and makes clear the distinction between the time-symmetric (TS) and time-forward (TF) approaches.

The basic ingredients of standard operational probabilistic theory are closed laboratories localized in space and time which take physical systems as input, perform local operations, and output new physical systems. The reader who is unfamiliar with quantum measurement theory (e.g. quantum channels, measurement operations, the Choi-Jamiołkowski isomorphism) may wish to consult Appendix~\ref{meas} for a brief introduction. The closed laboratory makes some classical information called an ``outcome'' available after the operation, which may in some instances be interpreted as a measurement outcome. In the TSOPT developed in \cite{Hardy:2021fqs}, an additional classical variable called an ``income''---the time-reversed counterpart of an outcome---is available before the operation is performed. An income may be interpreted as the initial state of a measuring apparatus or the initial value of a classical ancilla. To remember the difference between inputs/outputs and incomes/outcomes, one may use the following mnemonic: `p' is for \emph{physical} and `c' is for \emph{classical}. One of these local laboratories may be illustrated diagrammatically as a box with various wires coming out of it, as in Fig.~\ref{tsbox}.
\begin{figure}[ht] 
        \includegraphics[width=0.55\columnwidth]{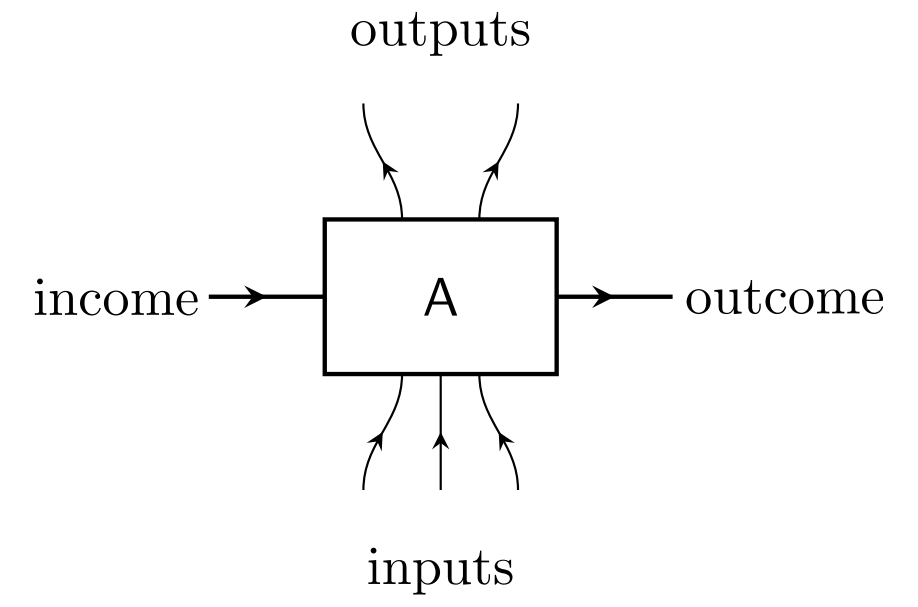} \centering
        \caption{
                \label{tsbox} 
                {Local laboratories in time-symmetric, operational probabilistic theory (TSOPT) are represented by a box with input and output wires carrying physical systems and income and outcome wires carrying classical information \cite{Hardy:2021fqs}.}           
        }
\end{figure}
In quantum theory, these boxes correspond to maps between Hermitian operators on input and output Hilbert spaces.

To illustrate the interpretation of the income variables, consider the following example, illustrated in Fig.~\ref{cartoon}.
\begin{figure}[ht] 
        {
        \includegraphics[width=1.0\columnwidth]{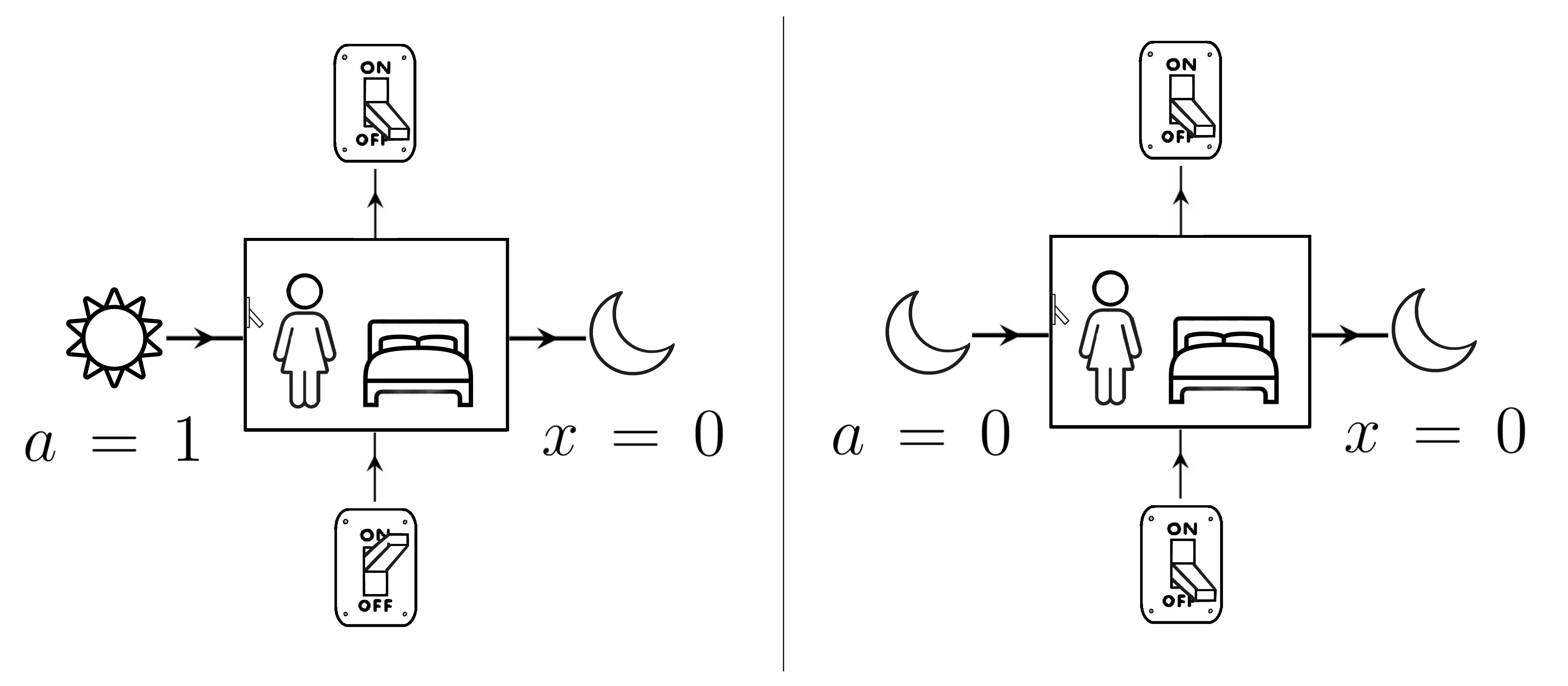} 
                \caption{An example of an operation illustrating the concept of income variables. Given Alice's choice of setting to turn off the light, Alice's action then depends on the value $a$ of the income. If the room is initially bright ($a=1$), Alice turns off the light, transforming the state of the light switch (the input/output), resulting in an outcome of darkness ($x=0$). If the room is initially dark ($a=0$), Alice acts trivially on the light switch, and the outcome is again darkness ($x=0$).}   
                \label{cartoon}
        }
\end{figure}
Alice wants to go to bed, so before she enters her bedroom, she decides she is going to turn off the light. Once she enters, she sees that the light is either on or off. If it is on, she flips the light switch off, and if it is off, she does nothing. Assuming that an external observer does not have access to Alice's memory of whether the light switch was initially on or off, the classical information of the initial state of the light switch is only available before Alice's operation is carried out. This is the hallmark feature of an income. Put differently, an external observer watching a video in reverse of Alice's actions only finds out the initial state of the light switch at the end of the reversed video. An income is distinguished from a \emph{setting}, which is some classical information that the agent, Alice, has free choice to determine at the time of the operation, independent of any external influences. In this example, the setting is Alice's decision to turn off the light and get ready for bed. Independent of the initial state of the light switch or any other factors, Alice has the agency to make any one of four possible decisions of how to operate on the light switch (two possible initial states times two final states equals four deterministic operations). In this example, the inputs and outputs are classical, but they could be quantum systems.

Boxes and wires can be joined together to form circuits, representing experiments carried out in a compact region of spacetime. For example, consider the circuit in Fig.~\ref{testpos}. Circuits with no loose wires represent the joint probability of obtaining the values of variables displayed in \emph{readout boxes} (e.g. the income and outcome variables $u,v,a,x$ in this example).
\begin{figure}[ht] 
        \includegraphics[width=0.7\columnwidth]{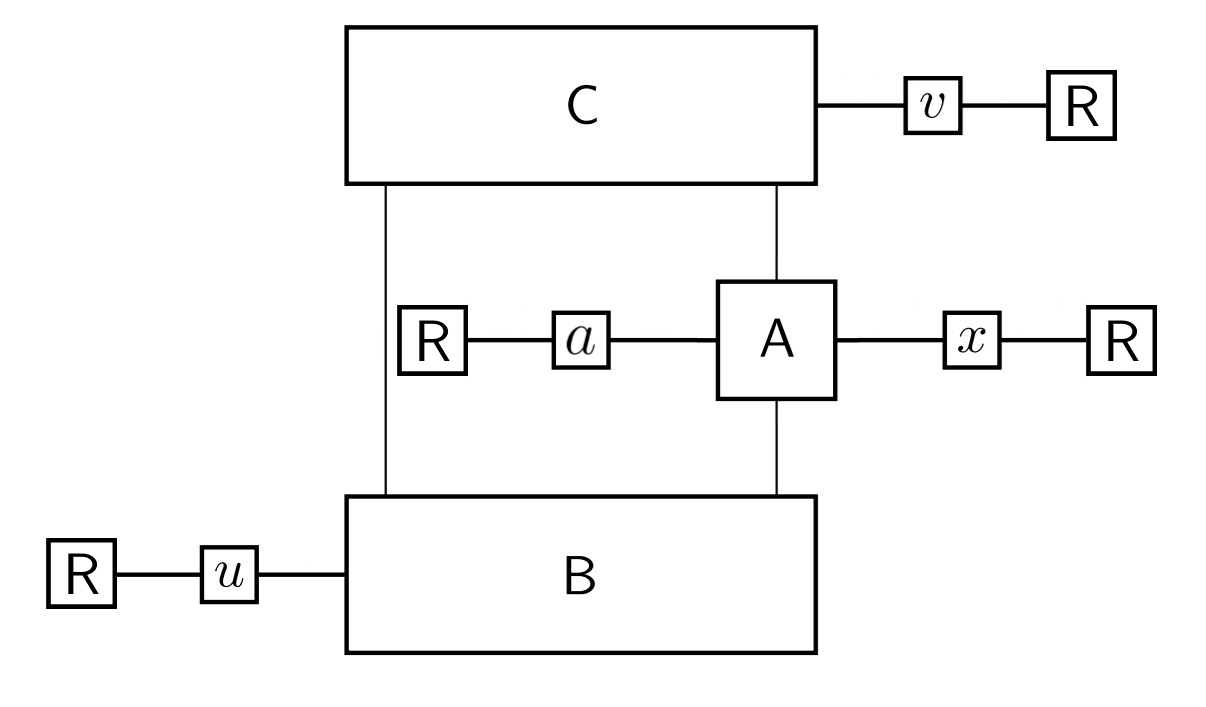} \centering
        \caption{
                \label{testpos} 
                {An example of a circuit in TSOPT. If the probability associated to this circuit is non-negative for all $B$ and $C$, the operation $A$ is said to be completely positive. This is the first physicality condition for operations in TSOPT.}           
        }
\end{figure}
This circuit consists of three operations labelled $\mathsf{A}$, $\mathsf{B}$, and $\mathsf{C}$, and represents the joint probability distribution $p(u,v,a,x)$ to observe the values $u,v,a,$ and $x$ in a single run of the experiment. In Ref.~\cite{Hardy:2021fqs} diagrammatic operator tensor notation (which originally appeared in Ref.~\cite{hardy2010formalismlocal}) was adapted to the time-symmetric setting.  The operator tensor notation looks, formally, like the diagrammatic operational notation (such as that used in Fig.~\ref{dubcausbox}).  While the operator tensor notation has certain advantages, in the present work we will use the Choi-Jamio\l kowski representation as used in Ref.~\cite{Oreshkov:2011er} since this notation is standard for work on process matrices.  

Different input/output wires in circuits may carry different types of physical systems. Additionally, different income/outcome wires may carry classical variables which take value in sets of different cardinality. We denote the cardinality of a classical variable, say $x$ for instance, by $N_x$. The boxes labeled with an `$\mathsf{R}$' indicate the transmission of random information such that each possible value for the classical variable is equally probable. In particular, a readout box $x$ sandwiched between two $\mathsf{R}$ boxes (see Fig.~\ref{readsan}) results in a circuit with a constant probability $1/N_x$. We also have the identity that summing, or \emph{marginalizing}, over the variable in a readout box produces a wire with nothing on it (Fig.~\ref{sumread}).
\begin{figure}[ht] 
        \includegraphics[width=0.33\columnwidth]{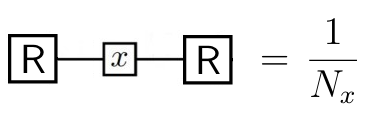} \centering
        \caption{
                \label{readsan} 
                {A readout box sandwiched between two $\mathsf{R}$ boxes gives a constant probability of one over the cardinality of the classical variable.}           
        }
\end{figure}
\begin{figure}[ht] 
        \includegraphics[width=0.39\columnwidth]{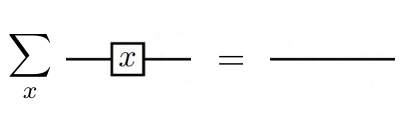} \centering
        \caption{
                \label{sumread} 
                {Marginalizing over the variable in a readout box produces a wire with nothing on it.}           
        }
\end{figure}

If the corresponding probability $p(u,v,a,x)$ of Fig.~\ref{testpos} is non-negative for all values of $u,v,a$ and $x$ and all boxes $\mathsf{B}$ and $\mathsf{C}$, this is equivalent to the complete positivity of operation $\mathsf{A}$. This complete positivity condition can be imposed in quantum theory by placing conditions on the operator associated with the operation $\mathsf A$. In quantum theory, the complete positivity condition of an operation $M^{A_IA_O}_{a,x}$ in the Choi-Jamiołkowski (CJ) representation can be written in operator language
\begin{equation} \label{cp}
    p(u,v,a,x) = \Tr_{A_IA_OB}[M^{A_IA_O}_{a,x}\vdot M^{A_IB}_u \vdot M^{A_OB}_v] \geq 0 \qquad  \forall M^{A_IB}_u, M^{A_OB}_v.
\end{equation}
Complete positivity is the first physicality constraint for operations in TSOPT. The second physicality constraint is called \emph{double causality} and is illustrated in Fig.~\ref{dubcausbox}.
\begin{figure}[ht] 
        \makebox[\textwidth][c]{        \includegraphics[width=1.1\textwidth]{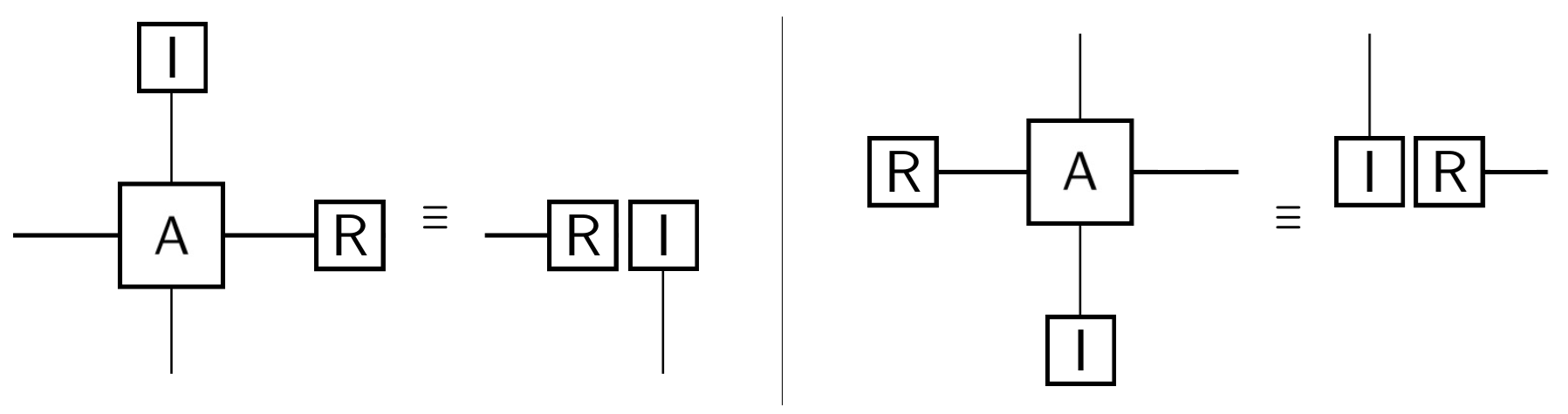}}
        \caption{
                \label{dubcausbox} 
                {Physical operations satisfy two \emph{double causality} conditions: forward causality (left) and backward causality (right). Open wires correspond to open wires on the left- and right-hand sides of an equality.}           
        }
\end{figure}
The box denoted with an `$\mathsf{I}$' is the ignore operation. In quantum theory, the `$\mathsf{I}$' box is taken to be the identity operator so that the corresponding Hilbert space is effectively ``ignored'' through a partial trace when calculating probabilities. In the operator language of quantum theory, the double causality constraints can be written,
\begin{equation} \label{dubcaus}
    \Tr_{A_O}\sum_x M_{a,x}^{A_IA_O} = \mathbb{1}^{A_I}, \qquad \frac{1}{N_a}\Tr_{A_I}\sum_aM_{a,x}^{A_IA_O} = \frac{1}{N_x}\mathbb{1}^{A_O}.
\end{equation}
We call the first condition here \emph{forward causality}, and the second \emph{backward causality}. These constraints are related to trace preservation, since if outcomes (incomes) are averaged over, the trace of the output (input) state is equal to the trace of the input (output) state. While these constraints may appear to be time-asymmetric due the $N_a$ and $N_x$ factors, there is a freedom in the normalization of the `$\mathsf{R}$' and `$\mathsf{I}$' boxes that can be used to shift these factors between the two constraints.  In the operator tensor notation as used in \cite{Hardy:2021fqs} it is possible to absorb these factors into the corresponding symbols such that the double causality constraints are in time-symmetric form.  In the current presentation, the placement of these factors was chosen to be the most intuitive for a reader familiar with the time-forward perspective. A quantum operation $M^{A_IA_O}_{a,x}$ satisfying Eqn.~(\ref{cp}) and Eqn.~(\ref{dubcaus}) is termed \emph{physical}. Note that the physicality constraints Eqns.~(\ref{cp}, \ref{dubcaus}) are manifestly time-symmetric. By restricting to circuits composed of physical operations, we are guaranteed to be dealing with processes that may equally well be run forward or backward in time. 

In Section~\ref{causineq}, we use the language of TSOPT to derive conditions that are satisfied by any process with a definite causal order. These conditions result in a set of forward and backward causal inequalities---a larger set than those considered previously by Branciard et al. \cite{Branciard_2016}. In Section~\ref{ts process}, we discuss the set of constraints that characterizes the process matrices consistent with time-symmetric quantum theory. We then expand the most general time-symmetric process matrix in a basis of Hermitian operators in order to analyze new types of allowed processes. In Section~\ref{discussion}, we present an example of a process matrix that violates a backward causal inequality but not a forward causal inequality. Hence, the time-symmetric process matrix formalism offers new opportunities to certify causal non-separability through backward causal inequality violation. We also show that the time-symmetric process matrix framework naturally incorporates processes with an indefinite time direction by explicitly constructing the process matrix associated to the so-called ``quantum time flip'' \cite{Chiribella:2020yfr, stromberg2022experimental, liu2023quantum}. The formalism presented here thus offers a unified framework for indefinite causal order and indefinite time direction. Further, we discuss in this section the one-to-one correspondence between operations in the time-symmetric (TS) and time-forward (TF) theories. This correspondence is broken in the presence of process matrices, due to the possibility for post-selection in the TS theory. Finally, we conclude in Section~\ref{conclude}.

\section{Processes with definite causal structure}
\label{causineq}

We begin by studying circuits that correspond to a process with definite causal order. Consider two parties, Alice and Bob, performing operations satisfying the time-symmetric physicality conditions of Fig.~\ref{testpos} and Fig.~\ref{dubcausbox}. Note that in this section we use TSOPT without restricting to quantum theory in order to remain completely general. 

The most general operational circuit with the definite causal ordering $A\preceq~\hspace{-0.2cm}B$ (Bob does not precede Alice) is given in Fig.~\ref{abprocess}.
\begin{figure}[ht] 
        \includegraphics[width=0.67\columnwidth]{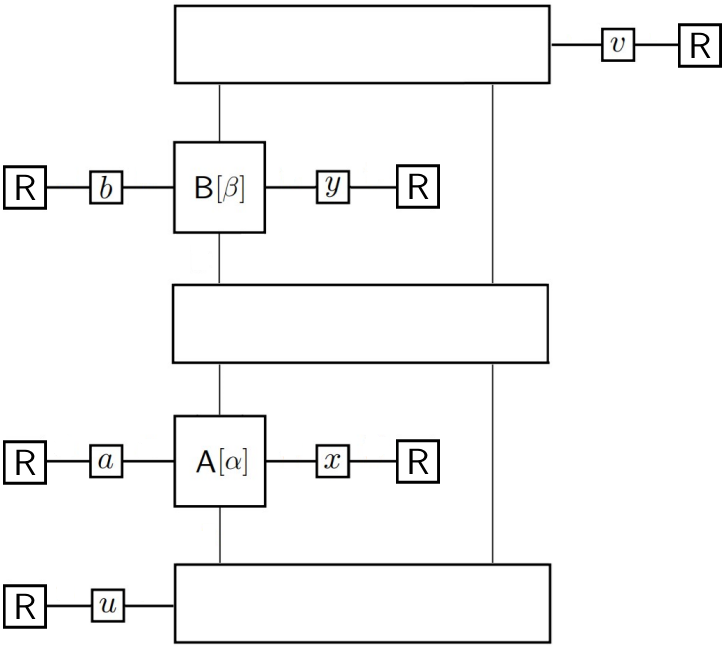} \centering
        \caption{
                \label{abprocess} 
                {The most general process where $B$ is in the causal future of $A$ ($A \preceq B$).}           
        }
\end{figure}
Here, we allow Alice and Bob to have setting choices $\alpha$ and $\beta$, respectively \cite{Hardy:2021fqs}. Each box is taken to be physical. For this reason, we make explicit the possibility for pre-selection and post-selection by including the variables $u$ and $v$, respectively. If the pre-selection variable $u$, for instance, is marginalized over, then by Eqn.~(\ref{dubcaus}) the lowermost box in the circuit is required to be the identity operation. We would like to allow for the possibility that a non-trivial physical system, such as a density matrix that is not maximally mixed, is known to be available at the beginning of the experiment. This necessitates the presence of the pre-selection variable $u$. In the spirit of time-symmetry, we similarly include the post-selection variable $v$. Additionally, we allow for the possibility of ancillary systems that may be entangled with Alice and Bob's inputs/outputs.

The circuit Fig.~\ref{abprocess} corresponds to a probability distribution 
\begin{equation*}
    p^{A\preceq B}(a,b,x,y,u,v|\alpha,\beta).
\end{equation*}
We can apply the double causality rules to obtain linear constraints on these probabilities. Marginalizing over the post-selection variable $v$ and Bob's outcome $y$, the circuit reduces to that of Fig.~\ref{marg} (i). This computation goes through by applying the identity of Fig.~\ref{sumread}, followed by the identity Fig.~\ref{readsan}, and finally by applying double causality Fig.~\ref{dubcausbox} twice. Note that an `$\mathsf{I}$' box connected to multiple wires factors into individual `$\mathsf{I}$' boxes on each wire, just as an identity operator acting on a tensor product of Hilbert spaces does.
\begin{figure}[ht] 
        \makebox[\textwidth][c]{        \includegraphics[width=1.15\textwidth]{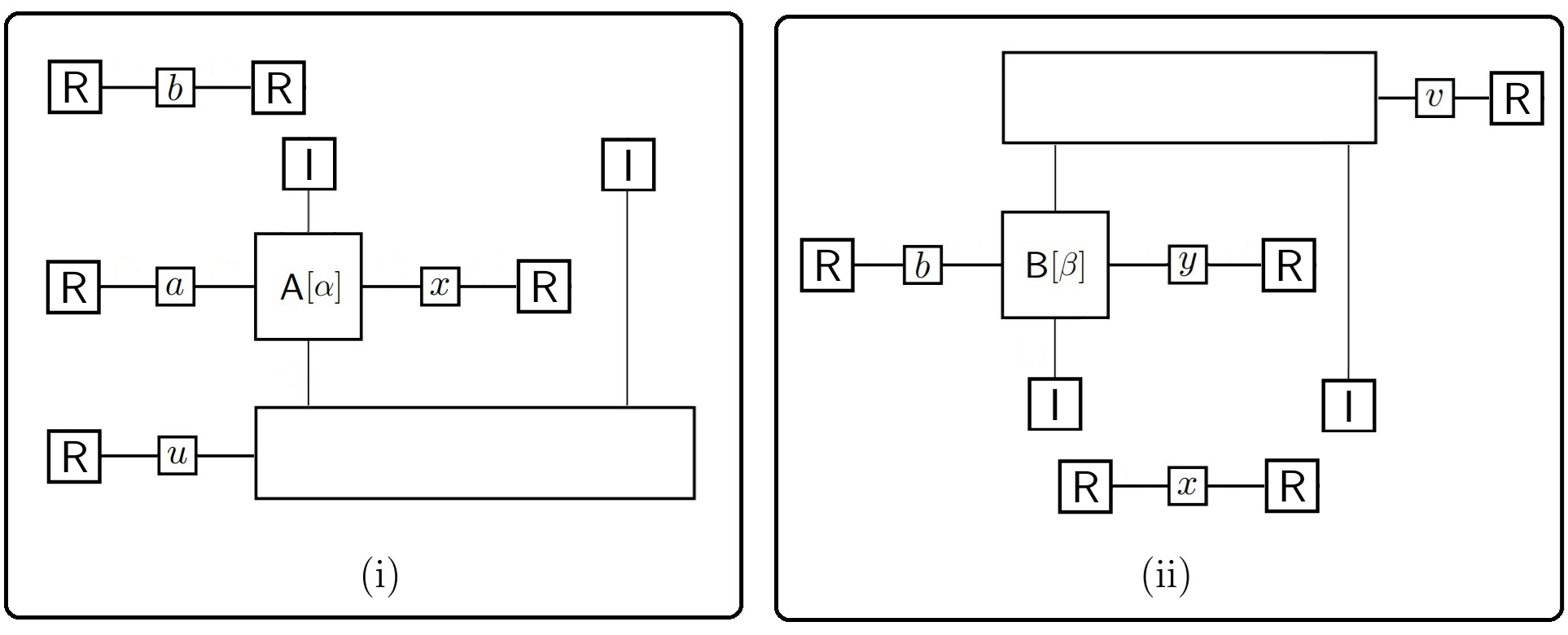}}
        \caption{
                \label{marg} 
                {The circuit of Fig.~\ref{abprocess} reduces by the double causality rules when certain variables are marginalized. (i) The variables $y$ and $v$ are marginalized. (ii) The variables $u$ and $a$ are marginalized.}           
        }
\end{figure}
The dependence on Bob's setting $\beta$ is made trivial, while the dependence on Bob's income $b$ factors out from the rest of the circuit. Thus we arrive at the familiar no-signalling constraint
\begin{equation} \label{ABforward}
    p^{A\preceq B}(a,b,x,u|\alpha,\beta) = \frac{1}{N_b}p^{A\preceq B}(a,x,u|\alpha), \qquad \forall a,b,x,u,\alpha,\beta.
\end{equation}
This equation states that Bob cannot signal to Alice unless there is post-selection---either in the post-selection variable $v$ or in Bob's outcome $y$. In other words, Eqn.~(\ref{ABforward}) states a statistical independence between Bob's setting $\beta$ and the variables $a, x, \alpha$ of Alice's laboratory. This is what is meant by no-signalling.

Now, we repeat the analysis by marginalizing over the pre-selection variable $u$ and Alice's income $a$. The circuit reduces in this case to that in Fig.~\ref{marg} (ii). The dependence on Alice's setting $\alpha$ and her outcome $x$ becomes trivial, resulting in an additional constraint
\begin{equation} \label{ABbackward}
    p^{A\preceq B}(b,x,y,v|\alpha,\beta) = \frac{1}{N_x}p^{A\preceq B}(b,y,v|\beta), \qquad \forall b,x,y,v,\alpha,\beta.
\end{equation}
This constraint is the time-reversal of Eqn.~(\ref{ABforward}), stating that Alice cannot signal to Bob unless there is pre-selection---either in the pre-selection variable $u$ or in Alice's income $a$. This inability to signal forward in time seems unfamiliar, yet, in real experiments there is almost always some form of pre-selection present, and hence the apparent issue is evaded. Whenever a measuring apparatus is prepared in some initial configuration, or the state of some physical system is known before the start of the experiment, there is pre-selection present and it is possible to signal forward in time.

It is straightforward to repeat the previous analysis for the definite causal ordering $B\preceq A$ (Alice does not precede Bob), illustrated in Fig.~\ref{BAprocess}.
\begin{figure}[ht] 
        \includegraphics[width=0.67\columnwidth]{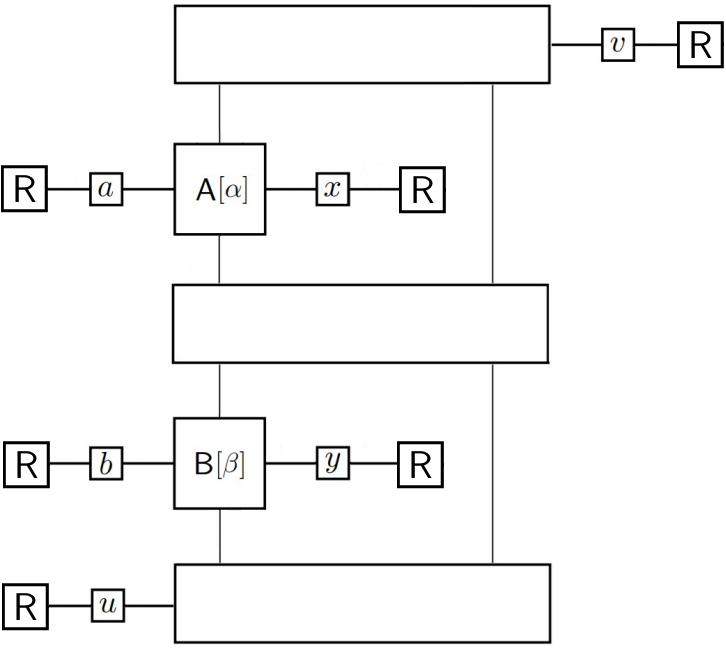} \centering
        \caption{
                \label{BAprocess} 
                {The most general process where $A$ is in the causal future of $B$ ($B \preceq A$).}           
        }
\end{figure}
We similarly arrive at two linear constraints:
\begin{equation} \label{BAforward}
    p^{B\preceq A}(a,b,y,u|\alpha,\beta) = \frac{1}{N_a}p^{B\preceq A}(b,y,u|\beta), \qquad \forall a,b,y,u,\alpha,\beta,
\end{equation}
stating no-signalling from Alice to Bob without post-selection, and
\begin{equation} \label{BAbackward}
    p^{B\preceq A}(a,x,y,v|\alpha,\beta) = \frac{1}{N_y}p^{B\preceq A}(a,x,v|\alpha), \qquad \forall a,x,y,v,\alpha,\beta,
\end{equation}
stating no-signalling from Bob to Alice without pre-selection.

Biparite \emph{causally separable} correlations are defined as those which are consistent with a definite causal ordering $A\preceq B$, $B \preceq A$, or a convex mixture of the two. That is,
\begin{align} \label{convex}
    p(a,b,x,y,u,v|\alpha,\beta) &= q\, p^{A\preceq B}(a,b,x,y,u,v|\alpha,\beta) \nonumber \\
    &\qquad + (1-q)\, p^{B\preceq A}(a,b,x,y,u,v|\alpha,\beta)
\end{align}
for some $q\in [0,1]$. Following the analysis of Branciard et al. Ref.~\cite{Branciard_2016}, we arrive at a set of causal inequalities that are necessarily satisfied by causally separable correlations. This derivation relies only on forward causality (the first condition in Fig.~\ref{dubcausbox}) and is therefore associated with the time-forward perspective. We adopt the terminology of Branciard et al. in naming the first of these GYNI (Guess Your Neighbor's Income):
\begin{equation} \label{GYNI}
    \frac{1}{N_aN_b}\sum_{a,b,x,y}\delta_{x,b}\,\delta_{y,a}\,p(x,y|\alpha,\beta,a,b,u,v) \leq \frac{1}{2}.
\end{equation}
This can be understood as placing an upper bound on the probability for success in a (time-forward) game where Bob's objective is to match his outcome to Alice's income, and Alice's objective is to match her outcome to Bob's income. The second causal inequality is called LGYNI (Lazy Guess Your Neighbor's Income):
\begin{equation} \label{LGYNI}
    \frac{1}{N_\alpha N_\beta N_a N_b}\sum_{\alpha,\beta, a,b,x,y}\delta_{\alpha(y\oplus a),0}\,\delta_{\beta(x\oplus b),0}\,p(x,y|\alpha,\beta,a,b,u,v) \leq \frac{3}{4}.
\end{equation}
The symbol `$\oplus$' denotes addition modulo 2. This second inequality corresponds to a modified game where a player is only required to guess their neighbor's income if their own setting is equal to one, otherwise they are free to produce any outcome they desire.

In the time-forward setting of Ref.~\cite{Branciard_2016}, this is an exhaustive list of the causal inequalities for two parties. The \emph{causal polytope} is the high-dimensional convex structure of causally separable probability distributions defined by Eqn.~(\ref{convex}) and no-signalling constraints \cite{gogioso2023geometry}. In the time-forward formalism, the GYNI and LGYNI causal inequalities describe all of the non-trivial facets of the causal polytope. There are more in the time-symmetric formalism. With time-symmetry, there is automatically a time-reversed counterpart for each of the causal inequalities. They are derived using only backwards causality (the second condition in Fig.~\ref{dubcausbox}) and are associated to the time-backward perspective. These include a time-reversed GYNI:
\begin{equation} \label{GYNIr}
    \frac{1}{N_xN_y}\sum_{a,b,x,y}\delta_{x,b}\,\delta_{y,a}\,p(a,b|\alpha,\beta,x,y,u,v) \leq \frac{1}{2},
\end{equation}
as well as a time-reversed LGYNI:
\begin{equation}\label{LGYNIr}
    \frac{1}{N_\alpha N_\beta N_x N_y}\sum_{\alpha,\beta, a,b,x,y}\delta_{\beta(y\oplus a),0}\,\delta_{\alpha(x\oplus b),0}\,p(a,b|\alpha,\beta,x,y,u,v) \leq \frac{3}{4}.
\end{equation}
We refer to the first two inequalities Eqns.~(\ref{GYNI}, \ref{LGYNI}) as the forward inequalities and the last two Eqns.~(\ref{GYNIr}, \ref{LGYNIr}) as the backward inequalities. Any causally separable process necessarily satisfies both the forward and backward inequalities. Note that Eqn.~(\ref{LGYNI}) and Eqn.~(\ref{LGYNIr}) are valid in these forms only up to two settings each, i.e. $1\leq N_\alpha, N_\beta \leq 2$. 

A violation of a causal inequality certifies the presence of indefinite causal structure in a given process. In Section~\ref{ts process}, we study correlations in time-symmetric quantum theory without assuming a definite causal order. We will then use the causal inequalities to certify that there exist causally non-separable processes that are consistent with time-symmetric quantum theory.

\section{Indefinite causal structure in time-symmetric quantum theory}
\label{ts process}

\subsection{The time-symmetric process matrix}
Consider a bipartite process where Alice and Bob act by quantum operations $M_{a,x}^{A_IA_O} \in \mathcal{L}(\mathcal{H}^{A_I}\otimes\mathcal{H}^{A_O})$ and $M_{b,y}^{B_IB_O}\in \mathcal{L}(\mathcal{H}^{B_I}\otimes\mathcal{H}^{B_O})$ in the CJ representation. We use the notation $\mathcal{L}(\mathcal{H}^X)$ to denote the space of linear operators on a Hilbert space $\mathcal{H}^X$. Without loss of generality, we temporarily omit the dependence on settings $\alpha$ and $\beta$. Alice has an income $a$ and an outcome $x$. Bob has an income $b$ and an outcome $y$. We require that Alice and Bob's operations are physical, satisfying completely positivitiy Eqn.~(\ref{cp}) and double causality Eqn.~(\ref{dubcaus}). In the time-symmetric scenario, input and output Hilbert spaces are of the same dimension which we denote by $d_A \coloneq d_{A_I} = d_{A_O}$ and $d_B \coloneq d_{B_I} = d_{B_O}$. 

The most general probabilities that can be associated to Alice and Bob's quantum operations are represented by the circuit in Fig.~\ref{prepostw}.
\begin{figure}[ht] 
        \includegraphics[width=1.0\columnwidth]{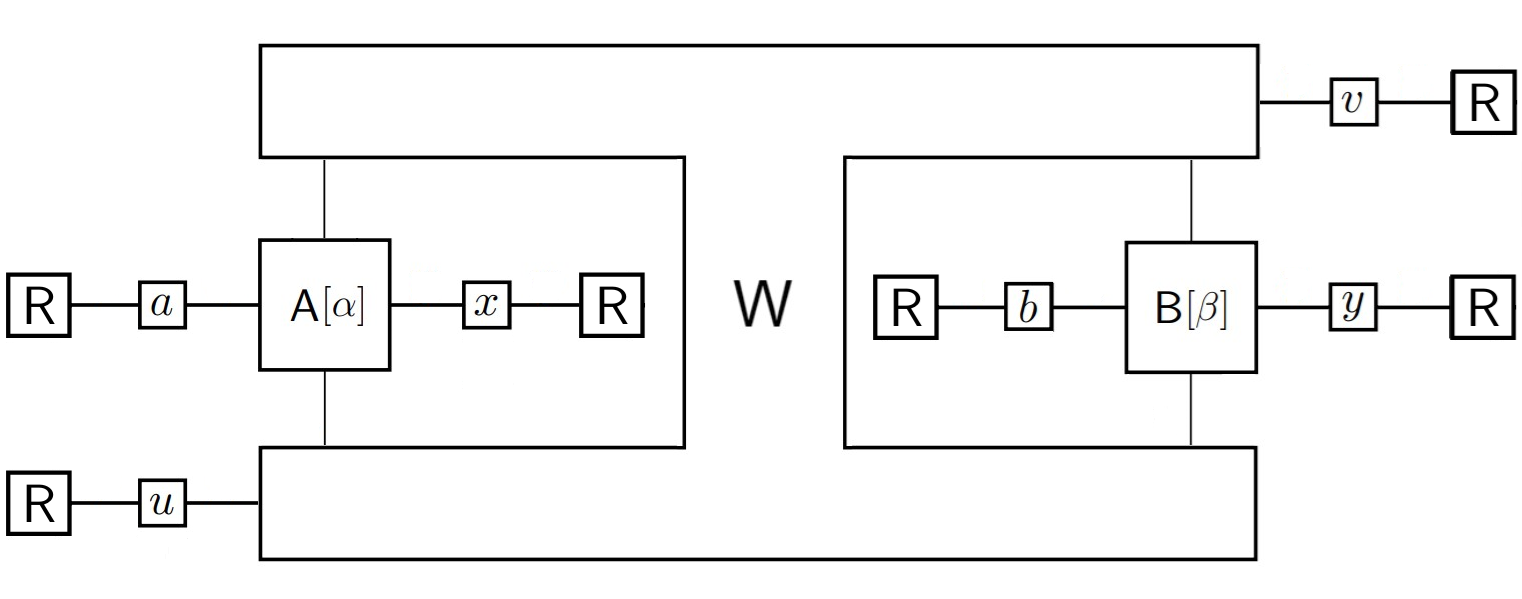} \centering
        \caption{
                \label{prepostw} 
                {The most general bipartite process involving local operations $A$ and $B$ is specified by a process matrix $W$.}           
        }
\end{figure}
The box labelled $W$ represents a Hermitian operator
\begin{equation}
    W_{u,v}^{A_IA_OB_IB_O}\in \mathcal{L}(\mathcal{H}^{A_I}\otimes\mathcal{H}^{A_O}\otimes\mathcal{H}^{B_I}\otimes\mathcal{H}^{B_O})
\end{equation}
called the \emph{process matrix} \cite{Oreshkov:2011er}. The classical variables $u$ ad $v$, which we call the \emph{pre-selection and post-selection variables}, represent information that is available before and after the experiment. The process matrix can be thought of as a generalization of a density matrix since it determines probabilities in an analogous way
\begin{equation}
    p(a,b,x,y,u,v) = \Tr_{A_IA_OB_IB_O}[W_{u,v}^{A_IA_OB_IB_O}\cdot(M_{a,x}^{A_IA_O}\otimes M_{b,y}^{B_IB_O}) ].
\end{equation}
More complicated objects can also be encoded in the process matrix, including quantum channels, quantum channels with memory, and, as we will see, indefinite causal structures. As opposed to earlier works \cite{Oreshkov:2011er} which implicitly condition on the pre-selection variable, here we study the joint probability $p(a,b,x,y,u,v)$ among all circuit variables, including the pre-selection variable $u$ and post-selection variable $v$. The transition between joint and conditioned probabilities is not straightforward in the present context since the normalization factor in Bayes' rule will depend non-trivially on Alice and Bob's choices of operations. An analogous phenomenon occurs in the context of the ABL rule \cite{Mohrhoff_2001}, where probabilities with pre-selection and post-selection in quantum theory contain a non-trivial normalization factor. When considering only pre-selection (or only post-selection) in the process matrix, the normalization factor is constant and it is trivial to transition between joint and conditioned probabilities. It will turn out that this is the case for the time-forward process matrices, but not the full set of time-symmetric process matrices.

In order to associate non-negative probabilities to circuits, we require that the process matrix is a positive semi-definite operator
\begin{equation} \label{positiveW}
    W_{u,v}^{A_IA_OB_IB_O} \geq 0 \qquad \forall u,v.
\end{equation}
We also require that the process matrix associates normalized probabilities. For simplicity, we define the averaged operations
\begin{equation} \label{avgop}
    M^{A_IA_O} \equiv \sum_{a,x}M_{a,x}^{A_IA_O}, \qquad M^{B_IB_O} \equiv \sum_{b,y}M_{b,y}^{B_IB_O}.
\end{equation}
Then the requirement of normalized probabilities is equivalent to imposing
\begin{align} \label{sumtoone}          \sum_{u,v}\Tr_{A_IA_OB_IB_O}\bigg[&W_{u,v}^{A_IA_OB_IB_O}\cdot(M^{A_IA_O}\otimes M^{B_IB_O}) \bigg]=1, 
\end{align}
for all physical, averaged operations $M^{A_IA_O}$ and $M^{B_IB_O}$. 

As shown in Appendix~\ref{derW}, the requirement of normalized probabilities translates into four linear constraints on the process matrix. We adopt the notation of Ref.~\cite{Branciard_2016} to write ${}_{X}W \coloneq \frac{1}{d_X}\mathds{1}^X\otimes\Tr_X[W]$ as the ``trace part'' of an operator $W\in \mathcal{L}(\mathcal{H}^X)$. We also use ${}_{[1-X]}W \coloneq W - {}_{X}W$ to denote the ``traceless part'' of $W$. Then the normalization constraints on the process matrix can be written: 
\begin{align}
    \sum_{u,v}\Tr_{A_IA_OB_IB_O}[W_{u,v}] &= d_Ad_B, \label{Wcons1}\\
    \sum_{u,v}{}_{B_IB_O[1-A_I][1-A_O]}W_{u,v}&=0, \label{Wcons2}\\
    \sum_{u,v}{}_{A_IA_O[1-B_I][1-B_O]}W_{u,v}&=0, \label{Wcons3}\\
    \sum_{u,v}{}_{[1-A_I][1-A_O][1-B_I][1-B_O]}W_{u,v}&=0. \label{Wcons4}
\end{align}
The first of these ensures that probabilities are normalized if Alice and Bob act trivially ($M^{A_IA_O} = (1/d_A)\mathds{1}^{A_IA_O}$), as illustrated in Fig.~\ref{normw}.
\begin{figure}[ht] 
        \includegraphics[width=0.73\columnwidth]{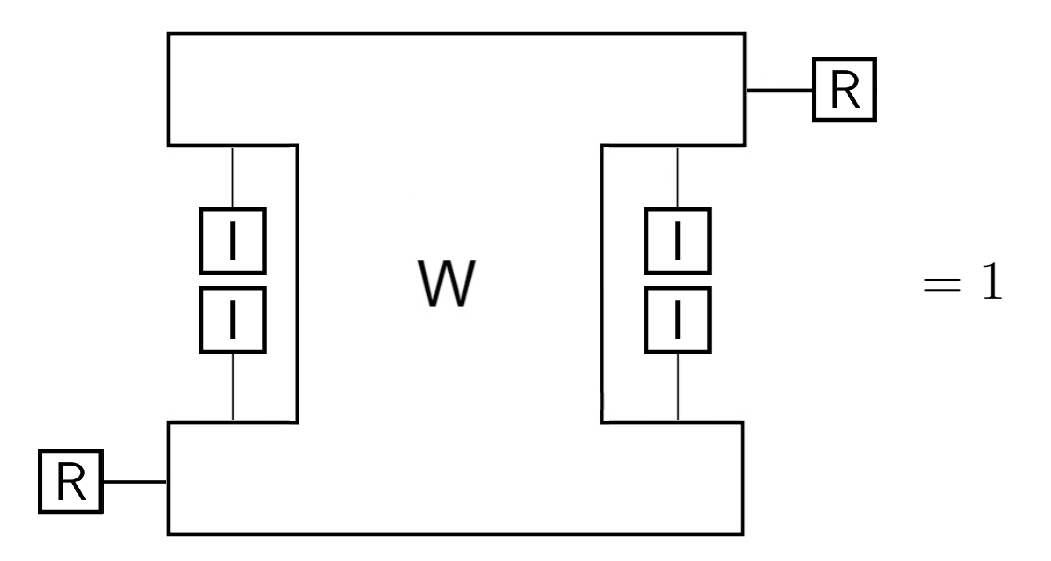} \centering
        \caption{
                \label{normw} 
                {Eqn.~(\ref{Wcons1}) requires that the process matrix $W$ associates normalized probabilities when Alice and Bob act with trivial operations. The readout boxes containing $u$ and $v$ are removed by the identity in Fig.~\ref{sumread}} once they are summed over.           
        }
\end{figure}
If Alice's and/or Bob's operation also contains a non-trivial traceless part, the remaining three constraints guarantee that probabilities remain normalized.

An additional set of constraints must be imposed on process matrices with no post-selection or no pre-selection, i.e. $u$ or $v$ is marginalized (summed), respectively. Process matrices with the post-selection variable marginalized satisfy additional no-signalling constraints:
\begin{align}
    \sum_v {}_{A_I[1-B_O]}W_{u,v} &= 0, \label{vcons1}\\
    \sum_v {}_{B_I[1-A_O]}W_{u,v} &= 0, \label{vcons2}\\
    \sum_v {}_{[1-A_O]}{}_{[1-B_O]}W_{u,v}&=0. \label{vcons3}
\end{align}
The derivation of these conditions is given in Appendix~\ref{nopostpre}. These three constraints guarantee that if one of the parties marginalizes over their outcome, then they cannot signal backward in time to the other party. This principle is respected by the no-signalling constraints derived in Section~\ref{causineq} and we assert that it remains true in indefinite causal structures. Similarly, there are three constraints for a process matrix with the pre-selection variable marginalized:
\begin{align}
    \sum_u {}_{A_O[1-B_I]}W_{u,v} &= 0, \label{ucons1}\\
    \sum_u {}_{B_O[1-A_I]}W_{u,v} &= 0, \label{ucons2}\\
    \sum_u {}_{[1-A_I]}{}_{[1-B_I]}W_{u,v}&=0. \label{ucons3}
\end{align}
These guarantee that if one party marginalizes over their income, they cannot signal forward in time to the other party.

The normalization conditions in Eqns.~(\ref{Wcons2}-\ref{Wcons4}) are redundant in light of the no-signalling conditions of Eqns.~(\ref{vcons1}-\ref{ucons3}). We include them anyway for completeness, in case one is interested which assumptions lead to which constraints, and for comparison with other works in the literature. Eqn.~(\ref{Wcons1}) is not redundant and must be imposed independently.

The process matrices with non-trivial pre-selection/post-selection (both $u$ and $v$ are not marginalized) form the largest class of bipartite process matrices. These process matrices satisfy positivity, Eqn.~(\ref{positiveW}), and must never assign probabilities greater than unity, but otherwise they are unconstrained. This is analogous to how an operation satisfies similar minimal constraints until either the outcome or income is marginalized, following which it must satisfy forward or backward causality. Chiribella and Liu study this same class of processes in Appendix D of Ref.~\cite{liu2024tsirelson}, referrring to them as processes with ``Indefinite Causal Order and Time Direction'' (ICOTD). They demonstrate that the classical processes with ICOTD form a sufficiently large set to produce any arbitrary probability distribution, even in the $N$-party case. In particular, they show that these classical ICOTD processes can achieve the algebraic maximum of every causal inequality.

\subsection{Expansion in basis operators}
The constraints for a process matrix with no pre-selection and/or no post-selection restrict the kinds of terms that are allowed to appear in the process matrix. It is instructive to expand in a basis set of Hermitian operators
\begin{equation}
    \{\sigma^X_\mu\}_{\mu=0}^{d_X^2-1}, \qquad \sigma_\mu^X \in \mathcal{L}(\mathcal{H}^X).
\end{equation}
It is always possible to choose a basis such that the $\mu=0$ operator is the identity,
\begin{equation}
    \sigma_0^X = \mathbb{1}^X,
\end{equation}
and all others are traceless,
\begin{equation}
    \Tr_X\sigma_j^X = 0, \qquad j>0.
\end{equation}
Further, we require this set of operators to be orthogonal under the Hilbert-Schmidt inner product
\begin{equation}
    \Tr_X[\sigma_\mu^X\sigma_\nu^X] = d_X \delta_{\mu\nu}.
\end{equation}
With these properties, any Hermitian operator on $\mathcal{H}^X$ can be decomposed as a linear combination of basis operators. In particular, an operation in the CJ representation has the decomposition
\begin{equation} \label{opdecomp}
    M^{X_IX_O} = \sum_{\mu\nu}\mathcal{X}_{\mu\nu}\sigma_\mu^{X_I}\sigma_\nu^{X_O} \in \mathcal{L}(\mathcal{H}^{X_I}\otimes\mathcal{H}^{X_O}).
\end{equation}
When there is no risk for confusion, we omit the tensor product $`\otimes$'. Then specifying the operation $M^{X_IX_O}$ amounts to specifying a set of coefficients $\mathcal{X}_{\mu\nu}$. The double causality constraints Eqn.~(\ref{dubcaus}) impose restrictions on the $\mathcal{X}_{\mu\nu}$,
\begin{equation}
    \mathcal{X}_{00} = \frac{1}{d_X}, \quad \mathcal{X}_{\mu0} = 0, \quad \mathcal{X}_{0\nu} = 0,
\end{equation}
as can be checked by taking the partial traces of Eqn.~(\ref{opdecomp}). Then we can write the most general, physical operation in the form
\begin{equation}
    M^{X_IX_O} = \frac{1}{d_X}\bigg( \mathbb{1}^{X_IX_O} + \sum_{ij>0}\mathcal{X}_{ij}\sigma_i^{X_I}\sigma_j^{X_O}\bigg).
\end{equation}
The coefficients $\mathcal{X}_{ij}$ with $i,j>0$ are free to take any values so long as the resulting $M^{X_IX_O}$ is positive semi-definite. Note that there are fewer allowed terms here in the time-symmetric formulation than were found in the time-forward case (Ref.~\cite{Oreshkov:2011er}, see Supplementary Methods). This is due to the fact that the double causality conditions contain an additional constraint (backward causality) that was absent previously in the time-forward formalism.

Now, we use the set of basis Hermitian operators to expand the bipartite process matrix
\begin{equation}
    W_{u,v}^{A_IA_OB_IB_O} = \sum_{\mu\nu\alpha\beta}w_{\mu\nu\alpha\beta}(u,v)\sigma_\mu^{A_I}\sigma_\nu^{A_O}\sigma_\alpha^{B_I}\sigma_\beta^{B_O},
\end{equation}
as well as the physical operations of Alice and Bob,
\begin{align}
     M^{A_IA_O} &= \frac{1}{d_A}\bigg( \mathbb{1}^{A_IA_O} + \sum_{ij>0}\mathcal{A}_{ij}\sigma_i^{A_I}\sigma_j^{A_O}\bigg), \nonumber \\
      M^{B_IB_O} &= \frac{1}{d_B}\bigg( \mathbb{1}^{B_IB_O} + \sum_{ij>0}\mathcal{B}_{ij}\sigma_i^{B_I}\sigma_j^{B_O}\bigg). \label{TSopexpand}
\end{align}
The coefficients $w_{\mu\nu\alpha\beta}(u,v)$ have a dependence on the pre-selection and post-selection variables which we will suppress from here on for simplicity of notation. Given these decompositions, we can rewrite the requirement of normalized probabilities in Eqn.~(\ref{sumtoone}),
\begin{align}
    1 &= \sum_{u,v}\Tr_{A_IA_OB_IB_O}[W_{u,v}^{A_IA_OB_IB_O}\cdot(M^{A_IA_O}\otimes M^{B_IB_O}) ] \nonumber \\
    &= \sum_{u,v}\bigg(d_Ad_Bw_{0000} + \frac{d_A}{d_B}\sum_{ij>0}w_{ij00}\mathcal{A}_{ij} + \frac{d_B}{d_A}\sum_{lm>0}w_{00lm}\mathcal{B}_{lm} \nonumber \\
    &\qquad\qquad\qquad\qquad\qquad + d_Ad_B\sum_{ijlm>0}w_{ijlm}\mathcal{A}_{ij}\mathcal{B}_{lm}\bigg)
\end{align}
using the defining properties of the set of orthogonal operators $\{\sigma_\mu^X\}_\mu$. The total probability is required to be unity for any choice of Alice and Bob's operations, that is, for any choice of $\mathcal{A}_{ij}$ and $\mathcal{B}_{ij}$. This results in the following requirements:
\begin{align}
    \sum_{u,v}w_{0000} = \frac{1}{d_Ad_B}, \quad \sum_{u,v}w_{ij00} = 0, \quad \sum_{u,v}w_{00lm} = 0, \quad &\sum_{u,v}w_{ijlm} = 0, \nonumber \\
    &\forall i,j,l,m>0.
\end{align}
We can further constrain the coefficients $w_{\mu\nu\alpha\beta}$ by enforcing Eqns.(\ref{vcons1}-\ref{ucons3}). The constraints corresponding to no post-selection can be written
\begin{equation} \label{nopostexp}
    \sum_v w_{0\alpha\beta i} = 0, \quad \sum_v w_{\alpha i 0 \beta} = 0, \quad \sum_v w_{\alpha i \beta j} = 0, \quad \forall \alpha, \beta \geq 0, \,i,j>0.
\end{equation}
The constraints corresponding to no pre-selection are
\begin{equation} \label{nopreexp}
    \sum_u w_{\alpha 0 i \beta} = 0, \quad \sum_u w_{i \alpha\beta 0} = 0, \quad \sum_u w_{i \alpha j \beta} = 0, \quad \forall \alpha, \beta \geq 0, \,i,j>0.
\end{equation}
The most general, physical bipartite process matrix can be written as a sum
\begin{equation} \label{genproc}
    W_{u,v}^{A_IA_OB_IB_O} = \frac{1}{d_Ad_B}\bigg(p_0(u,v)\mathbb{1}^{A_IA_OB_IB_O} + \sigma_{u,v}^{TS} + \sigma_{u,v}^{TF} + \sigma_{u,v}^{TB} + \sigma_{u,v}^{ISO} \bigg)
\end{equation}
where the operators $\sigma_{u,v}^{TS},  \sigma_{u,v}^{TF}, \sigma_{u,v}^{TB}, \sigma_{u,v}^{ISO}$ are Hermitian operators in $\mathcal{L}(\mathcal{H}^{A_I}\otimes\mathcal{H}^{A_O}\otimes\mathcal{H}^{B_I}\otimes\mathcal{H}^{B_O})$ which may depend non-trivially on $u$ and $v$. The factor $p_0(u,v)$ is the joint probability for observing the values $u,v$ with all other variables marginalized when Alice and Bob act trivially ($M^{A_IA_O} = (1/d_A)\mathds{1}^{A_IA_O}$). This of course must satisfy $\sum_{u,v}p_0(u,v) = 1$. The labels on the final four terms stand for ``Time-Symmetric,'' ``Time-Forward,'' ``Time-Backward,'' and ``Isolated,'' respectively. These operators are defined by:
\begin{align} \label{sigmaops}
    \sigma_{u,v}^{TS} &= \sum_{ij>0}\bigg(a_{ij}\sigma_i^{A_I}\sigma_j^{A_O} + b_{ij}\sigma_i^{B_I}\sigma_j^{B_O}\bigg) + \sum_{ijkl>0}c_{ijkl}\sigma_i^{A_I}\sigma_j^{A_O}\sigma_k^{B_I}\sigma_l^{B_O}, \nonumber \\
    \sigma_{u,v}^{TF} &= \sum_i\bigg(d_i\sigma_i^{A_I} + e_i\sigma_i^{B_I}\bigg) + \sum_{ij>0}f_{ij}\sigma_i^{A_I}\sigma_j^{B_I} \nonumber \\
    &\qquad\qquad\qquad\qquad + \sum_{ijk}\bigg(g_{ijk}\sigma_i^{A_I}\sigma_j^{A_O}\sigma_k^{B_I} + h_{ijk}\sigma_i^{A_I}\sigma_j^{B_I}\sigma_k^{B_O}\bigg), \nonumber \\
    \sigma_{u,v}^{TB} &= \sum_i\bigg(m_i\sigma_i^{A_O} + n_i\sigma_i^{B_O}\bigg) + \sum_{ij>0}o_{ij}\sigma_i^{A_O}\sigma_j^{B_O} \nonumber \\
    &\qquad\qquad\qquad\qquad + \sum_{ijk}\bigg(q_{ijk}\sigma_i^{A_I}\sigma_j^{A_O}\sigma_k^{B_O} + r_{ijk}\sigma_i^{A_O}\sigma_j^{B_I}\sigma_k^{B_O}\bigg), \nonumber \\
    \sigma_{u,v}^{ISO} &= \sum_{ij>0}\bigg(s_{ij}\sigma_i^{A_I}\sigma_j^{B_O} + t_{ij}\sigma_i^{A_O}\sigma_j^{B_I}\bigg).
\end{align}
Each of the 15 sets of coefficients appearing here---$a_{ij}, b_{ij}, c_{ijk}$, etc.---have a dependence on $u$ and $v$ that has not been written explicitly to save space. The first operator $\sigma_{u,v}^{TS}$ represents the terms that are only available with \emph{both pre-selection and post-selection}, since
\begin{equation}
    \sum_u \sigma_{u,v}^{TS} = \sum_v\sigma_{u,v}^{TS} = 0,
\end{equation}
as is required by Eqn.(\ref{nopostexp}) and Eqn.~(\ref{nopreexp}). The second operator $\sigma_{u,v}^{TF}$ represents the terms that are available in the time-forward picture, but not in the time-backward picture since they vanish when there is no pre-selection:
\begin{equation}
    \sum_u\sigma_{u,v}^{TF} = 0.
\end{equation}
This is required by Eqn.~(\ref{nopreexp}). Likewise, $\sigma_{u,v}^{TB}$ represents the terms available in the time-backward picture, but not the time-forward picture since they vanish when there is no post-selection:
\begin{equation}
    \sum_v\sigma_{u,v}^{TB} = 0.
\end{equation}
This is required by Eqn.~(\ref{nopostexp}). The final operator $\sigma_{u,v}^{ISO}$ represents terms which may be found in an isolated process, that is, one that does not require pre-selection or post-selection. This may be interpreted as a process that does not rely on any external resources such as density matrices which are not maximally mixed, or non-maximal measurements (these concepts are time-reversal counterparts). These terms are always available even if one marginalizes over $u$ and/or $v$.

Thus, in the time-symmetric process matrix formalism presented here, four classes of processes are arranged naturally in a hierarchical structure, illustrated in Fig.~\ref{hier}. 
\begin{figure}[ht] 
        \includegraphics[width=0.85\columnwidth]{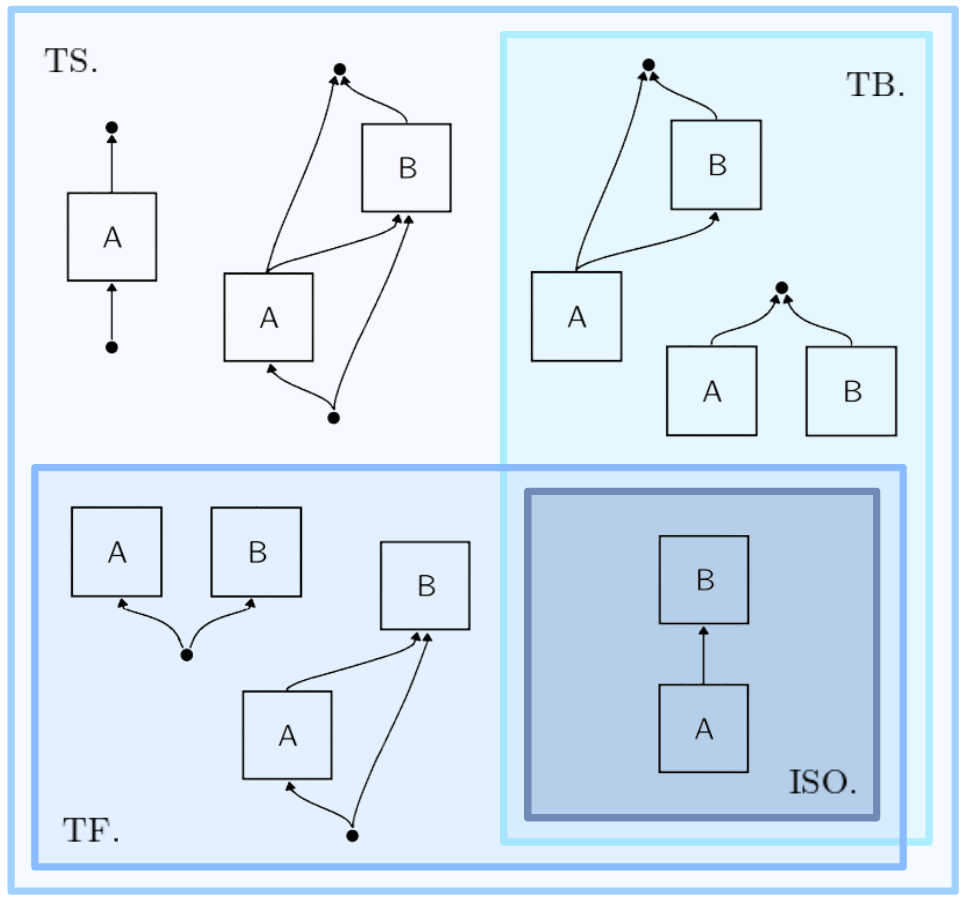} \centering
        \caption{
                \label{hier} 
                {Time-symmetric process matrices are naturally sorted into four classes, as explained in the text. Representative examples of processes with definite causal structure in each class are illustrated schematically. Dots represent pre-selection or post-selection.}           
        }
\end{figure}
Processes with both pre-selection and post-selection form the largest class (TS). These processes are only constrained to give probabilities which are non-negative and not greater than unity. Otherwise, these processes may contain any of the terms listed in Eqn.~(\ref{sigmaops}). One conclusion from this analysis is that the time-symmetric operations, which are more constrained than time-forward operations, result in less constrained process matrices. A consequence of the maximal size of the set of TS process matrices is that arbitrary bipartite probability distributions may be produced given the appropriate process and measurements. This is consistent with the results of Liu and Chiribella \cite{liu2024tsirelson} who studied this same class of processes from a different perspective. A sub-class (TF) is formed by the time-forward processes---those that do not involve post-selection. These processes may contain terms from $\sigma_{u,v}^{TF}$ and $\sigma_{u,v}^{ISO}$ and coincide with the known set of bipartite process matrices studied by Oreshkov, Costa, and Brukner \cite{Oreshkov:2011er}. Another sub-class (TB) is formed by the time-backward processes, which contain terms from $\sigma_{u,v}^{TB}$ and $\sigma_{u,v}^{ISO}$. Naturally, every time-backward process is the time-reversal of a time-forward process. Finally, the smallest sub-class (ISO) is formed by the isolated processes. These lie at the intersection of the time-forward and time-backward processes and contain terms only from $\sigma_{u,v}^{ISO}$. 

As can be seen from Fig.~\ref{hier}, a variety of processes with definite causal structure are possible in the time-symmetric formalism. The simplest examples are obtained by setting all coefficients in Eqn.~(\ref{sigmaops}) to zero apart from one of them. For example, a process like the example in the ISO category in Fig.~\ref{hier} can be obtained by choosing some of the $t_{ij}$ in Eqn.~(\ref{sigmaops}) to be non-zero. The interpretation of this process is a single quantum channel from Alice to Bob with no pre-selection or post-selection. A non-constant $p_0(u,v)$ indicates the presence of ancillary systems in the corresponding process (see, for example, Fig.~\ref{abprocess})---at least for causally separable processes. When Alice and Bob act trivially, there can only be correlation between $u$ and $v$ if there is some other ancillary channel through which signalling can occur. Choosing more complicated combinations of non-zero coefficients in Eqn.~(\ref{sigmaops}) can sometimes result in causally non-separable processes. We will see an example of this in Section~\ref{fbcausineq}. 

The physical interpretation of a causally non-separable process is not immediately clear, although it is commonly thought to be a result of quantum indefiniteness of some kind. This could be due to a spacetime metric in coherent superposition as in Ref.~\cite{moller2024gravitational}, where thin, spherical mass shells are put into a superposition of radii in order to implement the quantum switch \cite{vanderLugt:2022eqz}. Another possible mechanism for causally non-separable processes has been suggested by Oreshkov \cite{oreshkov2019time}. The proposal is based on the notion of time-delocalized quantum systems, which are ``nontrivial subsystems of the tensor products of Hilbert spaces associated with different times.'' In particular, Oreshkov has proposed time-delocalized systems as an explanation for realizations of the quantum switch in earthbound laboratories \cite{Rubino:2017upz, 2015NatCo...6.7913P}, since these experiments supposedly took place in definite near-Minkowski spacetime and cannot be due to spacetime indefiniteness.

\section{Results}
\label{discussion}
 
\subsection{Forward and backward causal inequalities}
\label{fbcausineq}

We have found that the known causal inequalities Eqns.~(\ref{GYNI}, \ref{LGYNI}) remain valid in the time-symmetric formalism. We have also found that the time-symmetric process matrix formalism results in all of the ICOTD processes considered by Chiribella and Liu \cite{liu2024tsirelson}. The ICOTD set of processes is known to be large enough to violate all causal inequalities.

An example of a causal inequality violating-process was studied previously by Oreshkov, Costa, and Brukner \cite{Oreshkov:2011er} where two parties deal only with qubits. This bipartite process is given by the following process matrix:
\begin{equation} \label{exampleW}
    W^{A_IA_OB_IB_O} = \frac{1}{4}\bigg[ \mathds{1}^{A_IA_OB_IB_O} + \frac{1}{\sqrt{2}}\bigg( \sigma_z^{A_O}\sigma_z^{B_I} + \sigma_z^{A_I}\sigma_x^{B_I}\sigma_z^{B_O}\bigg)\bigg],
\end{equation}
where $\sigma_z$ and $\sigma_x$ are Pauli matrices. There are implicit identity operators whenever the action on a particular Hilbert space is not specified within a term. One can check that the process matrix Eqn.~(\ref{exampleW}) satisfies the time-symmetric constraints Eqns.~(\ref{Wcons1}-\ref{Wcons4}) and is a positive semi-definite, Hermitian operator. This process matrix requires pre-selection, but not post-selection. This can be seen by comparing with the expansion in basis operators of Eqn.~(\ref{sigmaops}). We do not write any dependence on the pre-selection variable $u$ because there could be many ways to implement this particular process matrix with different way to depend on $u$.

Consider Alice's operation to consist of a measurement in the $z$ basis with outcome $x$ followed by a preparation of a state in the $z$ basis determined by the income $a$. Alice's operation in the CJ representation takes the form
\begin{align} \label{aliceop}
    M^{A_IA_O}_{a,x} = \frac{1}{4}[\mathds{1}+(-1)^x\sigma_z]^{A_I}\otimes [\mathds{1}+(-1)^a\sigma_z]^{A_O}.
\end{align}
Consider Bob to operate according to two settings. If $\beta=1$, Bob measures in the $z$ basis with outcome $y$ and prepares the maximally mixed state. If $\beta=0$, Bob measures in the $x$ basis with outcome $y$ and prepares a state in the $z$ basis determined by the income $b$, if $y=0$, or determined by $b\oplus 1$, if $y=1$. Overall, this is encoded in the CJ representation of Bob's operation
\begin{align} \label{bobop}
    M^{B_IB_O}_{b,y}[\beta] &= \frac{1}{2}\beta [\mathds{1}+(-1)^y\sigma_z]^{B_I}\otimes \mathds{1}^{B_O} \nonumber \\
    &\qquad + \frac{1}{4}(\beta\oplus 1)[\mathds{1}+(-1)^y\sigma_x]^{B_I}\otimes [\mathds{1}+(-1)^{b+y}\sigma_z]^{B_O}.
\end{align}
The operations of Eqn.~(\ref{aliceop}) and Eqn.~(\ref{bobop}) are completely positive and satisfy double causality. These were adopted from Ref.~\cite{Oreshkov:2011er} with one minor modification. In Ref.~\cite{Oreshkov:2011er}, Bob prepares an arbitrary state when $\beta=1$. To satisfy double causality, and not only forward causality, this state must be the maximally mixed state. This modification does not alter the resulting causal inequality violation. 

With these operations, Alice and Bob have a probability of success 
\begin{equation}
    p_\text{LGYNI} = \frac{2+\sqrt{2}}{4} > \frac{3}{4}
\end{equation}
in the forward LGYNI game that exceeds the bound of $3/4$ for causally separable processes. This is a violation of the forward LGYNI causal inequality Eqn.~(\ref{LGYNI}), certifying the presence of causal non-separability. Meanwhile, Alice and Bob's probability of success in the backward LGYNI game is
\begin{equation}
    \tilde{p}_\text{LGYNI} = \frac{1}{2} < \frac{3}{4},
\end{equation}
which does not violate the backward LGYNI causal inequality Eqn.~(\ref{LGYNIr}).

The time-symmetric process matrix formalism allows us to consider the time-reversal of Eqn.~(\ref{exampleW}),
\begin{equation}
    \tilde{W}^{A_IA_OB_IB_O} = \frac{1}{4}\bigg[ \mathds{1}^{A_IA_OB_IB_O} + \frac{1}{\sqrt{2}}\bigg( \sigma_z^{A_O}\sigma_z^{B_I} + \sigma_z^{A_I}\sigma_x^{A_O}\sigma_z^{B_O}\bigg)\bigg],
\end{equation}
and infer immediately that this too is a valid process matrix. To take the time-reversal of a process matrix, one must swap operators on the following Hilbert spaces: $A_O \leftrightarrow B_I$ and $A_I \leftrightarrow B_O$. This reverses inputs and outputs and also reverses the order of the two parties Alice and Bob. The last term in the reversed process matrix $\tilde{W}$ is type $A_IA_OB_O$ and is not allowed by the time-forward constraints of Ref.~\cite{Oreshkov:2011er}, but is allowed here in the time-symmetric formulation as long as there is post-selection. Using the time-reversals of Eqns.~(\ref{aliceop}, \ref{bobop}) as Alice and Bob's operations, which are guaranteed to be valid operations satisfying the time-symmetric constraints, results in a probability of success
\begin{equation}
    p_\text{LGYNI} = \frac{1}{2} < \frac{3}{4}
\end{equation}
in the forward LGYNI game. This is not a violation of the forward LGYNI inequality Eqn.~(\ref{LGYNI}). Therefore, it is not possible to certify the causal non-separability of $\tilde{W}$ with these operations using forward causal inequalities alone. However, Alice and Bob's probability of success in the backward LGYNI game is
\begin{equation}
    \tilde{p}_\text{LGYNI} = \frac{2+\sqrt{2}}{4} > \frac{3}{4}
\end{equation}
and consistutes a causal inequality violation. This example demonstrates that it is possible to violate a causal inequality and not its time-reversal, and that the backward causal inequalities Eqns.~(\ref{GYNIr}, \ref{LGYNIr}) make it possible to certify more processes as causally non-separable than was previously possible. 

It would be interesting in future research to determine whether there exists a process that can simultaneously violate forward and backward causal inequalities, and whether such a process offers any new computational advantages or physical insights. It would also be interesting to computationally generate an exhaustive list characterizing the facets of the causal polytope \cite{Giarmatzi2019} determined by Eqns.~(\ref{ABforward}-\ref{convex}). There may exist exotic causal inequalities beyond those presented here, perhaps some that cannot naturally be associated to a particular time direction but rather mix the forward and backward directions.

%%%%%%%%%%%%%%%%%%%%%%%%%%%%%%%%%%%%%%%%%%%%%%%%%%%

\subsection{The quantum time flip}
\label{quantumtf}

Chiribella and Liu characterized a class of time-symmetric operations called bidirectional devices in Ref.~\cite{Chiribella:2020yfr}. They demonstrated that it is possible in quantum theory to operate a bidirectional device in a coherent superposition of the forward time direction and backward time direction. Such a process is said to have \emph{indefinite time direction}. It has been claimed that the notion of an indefinite time direction could not be captured by the process matrix formalism \cite{Chiribella:2020yfr}. Here, we find that by manifestly incorporating time-symmetry into the process matrix framework it is indeed possible to describe processes with indefinite time direction. 

An example of a process with indefinite time direction is the \emph{quantum time flip}, illustrated in Fig.~\ref{qtf}. 
\begin{figure}[ht] 
        \includegraphics[width=0.56\columnwidth]{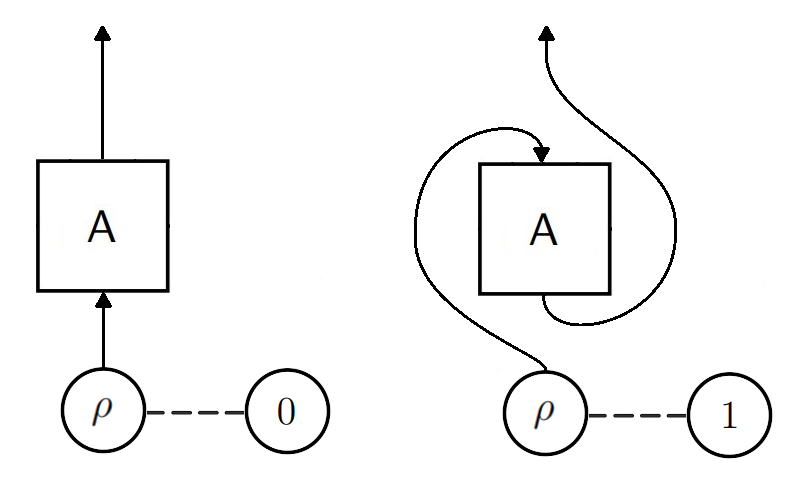} \centering
        \caption{
                \label{qtf} 
                {The quantum time flip can be implemented in the time-symmetric process matrix framework using $W_{QTF}$. A control qubit determines the time direction in which Alice's operation acts on a quantum state with density matrix $\rho$.}           
        }
\end{figure}
Alice performs an operation satisfying the time-symmetric conditions Eqn.~(\ref{cp}) and Eqn.~(\ref{dubcaus}). Meanwhile, Bob has access to a control qubit whose value determines the time direction in which Alice's operation is performed. If the control qubit reads 0, Alice's operation acts in the forward time direction, and if the control qubit reads 1, Alice's operation acts in the backward time direction. This scenario is described by a process matrix
\begin{align} \label{WQTF}
    W_{QTF} &= \bigg[\mathds{1}^{A_IA_O}\otimes \ket{0}{}^{B_I}\bra{0} + \text{SWAP}^{A_IA_O}\otimes \ket{1}{}^{B_I}\bra{1}\bigg] \nonumber \\
    &\qquad\vdot \rho^{A_IB_I}\vdot \bigg[\mathds{1}^{A_IA_O}\otimes \ket{0}{}^{B_I}\bra{0} + \text{SWAP}^{A_IA_O}\otimes \ket{1}{}^{B_I}\bra{1}\bigg]
\end{align}
satisfying the time-symmetric constraints Eqn.~(\ref{positiveW}) and Eqns.~(\ref{Wcons1}-\ref{Wcons4}). The operator $\rho$ is a density matrix (unit trace) that describes the composite system of the control qubit together with the quantum state Alice acts upon. Effectively, what $W_{QTF}$ does if the control qubit reads 0 is send Alice's component of $\rho$ into the input terminal of Alice's operation, and if the control qubit reads 1, it sends it into the output terminal of Alice's operation. This is always possible in the time-symmetric process matrix framework because operations are assumed to be valid in both time directions.

The operator $\text{SWAP}^{A_IA_O}$ is the usual unitary gate defined by
\begin{equation}
    \text{SWAP}^{A_IA_O} = \sum_{i,j=1}^{d_A}\ket{i}{}^{A_I}\bra{j}\otimes\ket{j}{}^{A_O}\bra{i}
\end{equation}
in the computational bases of $\mathcal{H}^{A_I}$ and $\mathcal{H}^{A_O}$. In the case of qubits ($\dim\mathcal{H}^{A_I} = \dim\mathcal{H}^{A_I} = 2$), this can be represented as a matrix
\begin{equation}
    \text{SWAP}^{A_IA_O} = \begin{pmatrix}
1 & 0 & 0 & 0 \\
0 & 0 & 1 & 0 \\
0 & 1 & 0 & 0 \\
0 & 0 & 0 & 1
\end{pmatrix}.
\end{equation}
There is some freedom here in how the computational bases of $\mathcal{H}^{A_I}$ and $\mathcal{H}^{A_O}$ are to be identified relative to one another. This freedom is hypothesized to be equivalent to the choice of the input-output inversion map $\Theta$ described by Chiribella and Liu in Ref.~\cite{Chiribella:2020yfr}. The process matrix $W_{QTF}$ implements the quantum time flip as described in Ref.~\cite{Chiribella:2020yfr}, as can be seen by writing the resulting correlations
\begin{equation}
    p_{QTF}(a,b,x,y) = \Tr_{A_IA_OB_IB_O}[W_{QTF}\vdot(M^{A_IA_O}_{a,x}\otimes M^{B_IB_O}_{b,y})]
\end{equation}
and expanding Alice and Bob's operations in Kraus operators. 

The time-symmetric process matrix framework offers a unified framework for studying processes with indefinite causal structure and indefinite time direction, synthesizing the approaches of Ref.~\cite{Oreshkov:2011er} and Ref.~\cite{Chiribella:2020yfr}. The time-symmetric constraints presented here characterize all processes incorporating indefinite causal structure and/or indefinite time direction consistent with time-symmetric quantum theory. 

%%%%%%%%%%%%%%%%%%%%%%%%%%%%%%%%%%%%%%%%%%%%%%%%%%%%%%%

\subsection{Time-symmetric versus time-forward}
With definite causal structure, time-symmetric (TS) and time-forward (TF) quantum theory are equivalent. This can be seen by the existence of a one-to-one correspondence between operations in the two theories. The fact that every TS operation has a counterpart in TF theory is straightforward to prove: every TS operation satisfies forward causality automatically, so simply conditioning on the income variable results in a collection of valid TF operations, illustrated in Fig.~\ref{TSTF}. The value of the income in the TS operation can be interpreted as a setting in the resulting collection of TF operations. Therefore, a TS operation corresponds to a particular collection of TF operations which all together satisfy backwards causality.
\begin{figure}[ht] 
        \includegraphics[width=0.58\columnwidth]{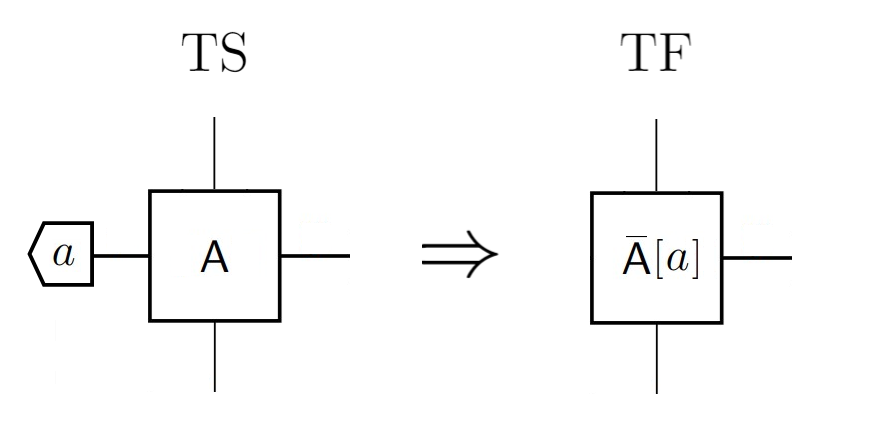} \centering
        \caption{
                \label{TSTF} 
                {Every TS operation corresponds to a collection of valid TF operations, obtained by conditioning on the income. The pentagon-shaped box represents conditioning on the variable $a$ using Bayes' rule.}           
        }
\end{figure}

The other direction of the proof is slightly more involved. It can be shown that, given a TF operation $\bar{M}^{A_IA_O}_a$, a TS operation can be constructed:
\begin{align}
    M^{A_IA_O}_{a,x} &= \frac{1}{N_a}\bar{M}^{A_IA_O}_x \delta_{a,1} \nonumber \\
    &\quad\quad + \frac{1}{d_A N_x (N_a-1)}\sum_{a' > 1}\mathds{1}^{A_I}\otimes\bigg(N_a\mathds{1}^{A_O} - \Tr_{A_I}\bar{M}^{A_IA_O}_x \bigg)\delta_{a,a'}.\label{TScons}
\end{align}
The TS operation is constructed so that pre-selecting on the income $a=1$ gives the TF operation, that is,
\begin{equation}
    \frac{1}{p(a=1)} M^{A_IA_O}_{1,x} = N_a M^{A_IA_O}_{1,x} = \bar{M}^{A_IA_O}_x.
\end{equation}
The TS operation $M^{A_IA_O}_{a,x}$ constructed in this way is guaranteed to be positive so long as the number of incomes satisfies $N_a\geq d_A N_x$, and it is guaranteed to satisfy double causality so long as $\bar{M}^{A_IA_O}_x$ satisfies forward causality. Therefore, we are free to interpret any TF operation as a TS operation conditioned on the income $a=1$.

Any circuit composed of TS operations is equivalent to a collection of circuits composed of TF operations obtained by pre-selecting on the possible combinations of income values. On the other hand, any circuit composed of TF operations is equivalent to a circuit composed of TS operations constructed by Eqn.~(\ref{TScons}) all pre-selected on the income value $a=1$. Thus, the same probability distributions can be obtained in both theories, and TF and TS quantum theory with definite causal structure are equivalent.

This equivalence fails to remain true when the assumption of definite causal structure is relaxed. As discussed in the text following Eqn.~(\ref{genproc}), the TS process matrices form a larger subspace in the space of linear operators $\mathcal{L}(\mathcal{H}^{A_I}\otimes\mathcal{H}^{A_O}\otimes\mathcal{H}^{B_I}\otimes\mathcal{H}^{B_O})$ than the TF process matrices. It is possible that one could construct an alternate formulation of TF quantum theory with indefinite causal structure that allows for post-selection in the process matrix. If only the process matrices averaged over the possible post-selection values are required to satisfy the usual TF process matrix constraints, it is possible an equivalence could be established between TS and TF theory. However, in their current forms, TS and TF quantum theory with indefinite causal structure are distinct. Given the same set of operations (according to the correspondence established in this section), there is a larger space of allowed probability distributions in TS theory than TF theory. 

To demonstrate this inequivalence concretely, consider the following simple scenario. To avoid developing the formalism for single-party process matrices, we present an example with two parties. However, the inequivalence between TS and TF process matrix theory manifests even with one party. Alice and Bob have inputs and outputs which are qubits, $d_A=d_B = 2$, and their operations are given as follows:
\begin{align}
    M^{A_IA_O}_{1,1} &= \frac{1}{4}(\mathds{1}^{A_IA_O} + c\sigma_z^{A_O}), \nonumber \\
    M^{A_IA_O}_{2,1} &= \frac{1}{4}(\mathds{1}^{A_IA_O} - c\sigma_z^{A_O}), \nonumber \\
    M^{B_IB_O}_{1,1} &= \frac{1}{2}\mathds{1}^{B_IB_O},
\end{align}
where $0\leq c\leq 1$ is a free parameter and $\sigma_z$ is the Pauli $z$ operator. Alice's TS operation corresponds to a TF operation
\begin{equation}
    \bar{M}^{A_IA_O}_{1} = \frac{1}{2}(\mathds{1}^{A_IA_O} + c\sigma_z^{A_O})
\end{equation}
by the construction of Eqn.~(\ref{TScons}) with $N_a=2$. Bob's operation is both a valid TS and TF operation. The process matrix is
\begin{equation}
    W = \frac{1}{4}(\mathds{1}^{A_IA_OB_IB_O}+w\sigma_z^{A_O})
\end{equation}
where $0< w\leq 1$ is a constant. $W$ is a valid process matrix only in the TS or TB theories, and is only permitted if post-selection is present. Then we can calculate
\begin{equation}
    p(1,1,1,1) \equiv \Tr[W\vdot (M^{A_IA_O}_{1,1}\otimes M^{B_IB_O}_{1,1})] = \frac{1}{2}(1+cw).
\end{equation}

Thus it is possible to have non-trivial dependence on the parameter $c$ in TS theory, however, this is impossible in TF theory. This is due to the fact that the term resulting from $c\sigma_z^{A_O}$ in Alice's operation is orthogonal to every possible term in a TF process matrix. Given this set of operations, which are valid in both theories, it is possible to demonstrate concretely the inequivalence between TS and TF process matrix theory in their current forms.

\section{Conclusion}
\label{conclude}

In this paper, we have developed a formalism for causal inequalities and process matrices in the setting of time-symmetric operational theory. We use a modified type of operation which includes \emph{incomes} as well as outcomes and satisfies a time-symmetric set of constraints. By studying processes with definite causal structure, we arrived at twice as many causal inequalities as those previously known for two parties. Each causal inequality from the time-forward setting (GYNI and LGYNI) has a time-reversed counterpart in our formalism, Eqns.~(\ref{GYNIr}, \ref{LGYNIr}). We demonstrated in Section~\ref{fbcausineq} that this larger set of causal inequalities offers new opportunities to certify the causal non-separability of certain processes which violate one of the backward inequalities. It remains to be shown whether these four causal inequalities form an exhaustive list for two parties, or whether there may exist additional inequalities. 

By requiring non-negative and normalized probabilities, while allowing for both pre-selection and post-selection, we derived the largest possible set of process matrices for two parties. This maximal set corresponds to the ICOTD processes of Chiribella and Liu \cite{liu2024tsirelson}, which were shown to maximally violate every causal inequality. In our formalism, process matrices without post-selection satisfy an additional set of constraints to prevent certain types of backwards signalling, see Eqns.~(\ref{vcons1}-\ref{vcons3}). The process matrices satisfying these constraints are precisely those of the TF setting \cite{Oreshkov:2011er}. This demonstrates that the primary distinction between process matrices in the TS and TF settings is the possibility for post-selection.

With post-selection permitted, the set of process matrices in the TS formalism (or ICOTD processes) contain certain processes with \emph{indefinite time direction}. One example of such a process is the quantum time flip, where the time direction of Alice's operation is determined by an ancillary control qubit, which may be prepared in a coherent superposition. The process matrix for the quantum time flip is written explicitly in Section~\ref{quantumtf}. A question for future research is to determine whether processes with indefinite time direction violate any causal inequalites that cannot be violated by a process with definite time direction. 

It would be interesting to find a process (a process matrix together with a set of operations) that violates both the forward and backward versions of a causal inequality. If there exists a normalized probability distribution which violates both the forward and backward versions of an inequality, then by the results of Chiribella and Liu \cite{liu2024tsirelson} we are guaranteed to be able to find an ICOTD process which produces this probability distribution. We expect that if such a process exists, it will involve indefinite time direction. This process will therefore exist in the full set of time-symmetric processes, but not the sub-class of time-forward processes (see Fig.~\ref{hier}).

\section*{Acknowledgments}
We would like to thank \v{C}aslav Brukner, Giulio Chiribella, Djordje Minic, Ognyan Oreshkov, and Aldo Riello for helpful discussions. Research at Perimeter Institute is supported in part by the Government of Canada through the Department of Innovation, Science and Economic Development Canada and by the Province of Ontario through the Ministry of Colleges and Universities.

\appendix

\section{Fundamentals of quantum measurement theory}
\label{meas}

This Appendix offers a brief introduction to quantum measurement theory for the unfamiliar reader. For a more comprehensive introduction, we refer the reader to the standard textbook on quantum information and quantum computation by Nielsen and Chuang \cite{nielsen2001quantum}. In the main text, we will rely primarily on the notion of an operation in the Choi–Jamiołkowski (CJ) representation. Here, the more commonly known time-forward formulation is presented, while in the main text we discuss the adaptation to the time-symmetric setting.

A \emph{quantum channel} takes a quantum system, represented by a density matrix in the finite-dimensional case, and maps it to a new quantum system. Mathematically, this corresponds to a completely positive, trace-preserving (CPTP) map. If the input and output quantum systems are represented by density matrices $\rho^I \in \mathcal{L}(\mathcal{H}^{I}), \rho^O \in \mathcal{L}(\mathcal{H}^{O})$ acting on their respective Hilbert spaces, a quantum channel is a CPTP map
\begin{align}
    \mathcal{C}: \mathcal{L}(\mathcal{H}^{I}) &\rightarrow \mathcal{L}(\mathcal{H}^{O}) \nonumber \\
    \rho^I &\mapsto \rho^O.
\end{align}
between the sets of Hermitian, linear operators on two Hilbert spaces. Complete positivity is the requirement that positive operators are mapped to positive operators, even in the presence of ancillary systems. This requirement is equivalent to the positivity of the induced map
\begin{equation}
    \mathds{1}^k \otimes \mathcal{C}: \mathcal{L}(\mathcal{H}^{A}\otimes\mathcal{H}^{I}) \rightarrow \mathcal{L}(\mathcal{H}^{A}\otimes\mathcal{H}^{O}) 
\end{equation}
for all $k$, where $\mathcal{H}^A$ is the Hilbert space of an auxiliary quantum system and $\dim{\mathcal{H}^A} = k$ is its dimension. Trace-preservation is the requirement that 
\begin{equation}
    \Tr_O [\mathcal{C}(\rho^I)] = 1
\end{equation}
for all normalized density matrices $\rho^I \in \mathcal{H}^I$. Complete positivity and trace-preservation guarantee that density matrices are mapped to density matrices under the action of the quantum channel.

We may loosen the requirement of trace-preservation and consider maps which are trace non-increasing, that is,
\begin{equation}
    0 \leq \Tr_O[\mathcal{C}(\rho^I)] \leq 1
\end{equation}
for all unit-normalized density matrices $\rho^I \in \mathcal{H}^I$. Trace non-increasing channels offer a natural description of quantum measurements. We may label a collection of trace non-increasing channels with subscripts, $\mathcal{C}_x$, where $x$ is the \emph{outcome} of a given measurement. Then the probability associated to the measurement outcome $x$ is
\begin{equation}
    p(x) \coloneq \Tr_O[\mathcal{C}_x(\rho^I)].
\end{equation}
These probabilities are guaranteed to be non-negative so long as $\rho^I$ is a valid density matrix. To guarantee normalized probabilities, that is,
\begin{equation}
    \sum_x p(x) = \sum_x \Tr_O[\mathcal{C}_x(\rho^I)] = 1,
\end{equation}
we require that the channel $\sum_x\mathcal{C}_x$, defined linearly, resulting from \emph{marginalization} over the outcome $x$ is a CPTP channel. We refer to the collection $\{\mathcal{C}_x\}_x$ as an \emph{operation} to indicate that it represents a set of generalized measurements, that is, trace non-increasing channels with associated outcomes. This is not to be confused with the more familiar \emph{projective measurements}---operations with the additional requirement that the constituent channels act projectively on quantum states. We will see in Section~\ref{intro} of the main text that the definition of an operation is modified in the time-symmetric setting to include \emph{incomes} in addition to outcomes. This is explained in detail in the main text.

A useful representation of an operation $\{\mathcal{C}_x\}_x$ is given by the \emph{Choi–Jamiołkowski (CJ) isomorphism}. To construct the isomorphism, consider the (un-normalized) maximally entangled state
\begin{equation}
    \ket{\Phi^+} = \sum_{i=0}^{\dim (\mathcal{H}^I)-1}\ket{i}\otimes\ket{i} \in \mathcal{H}^I\otimes\mathcal{H}^I.
\end{equation}
on two copies of the input Hilbert space $\mathcal{H}^I$ where $\{\ket{i} \in \mathcal{H}^I\}$ forms an orthonormal basis. Then we can construct
\begin{equation}
    M^{IO}_x := \mathds{1}^I\otimes\mathcal{C}_x\bigg(\ket{\Phi^+}\bra{\Phi^+}\bigg) \in \mathcal{L}(\mathcal{H}^I\otimes\mathcal{H}^O).
\end{equation}
This is the CJ representation of the operation $\{\mathcal{C}_x\}_x$. The complete positivity of $\mathcal{C}_x$ is equivalent to the positive semi-definiteness of $M^{IO}_x$, and trace-preservation becomes the condition
\begin{equation}
    \sum_x \Tr_O[M^{IO}_x] = \mathds{1}^I.
\end{equation}
We call this condition \emph{forward causality}. Again, in the time-symmetric scenario the definition of an operation is modified, and we are instead left with two constraints which we call collectively \emph{double causality}, see Eqn.~(\ref{dubcaus}). The CJ representation of time-symmetric operations is the main ingredient in our discussion of quantum theory throughout this paper.

\section{Derivation of time-symmetric process matrix constraints}
\label{derW}

\begin{theorem} \label{wconstraints}
The requirement of normalized probabilities as in Eqn.~(\ref{sumtoone}) is satisfied iff the process matrix $W_{u,v}^{A_IA_OB_IB_O}$ satisfies the following four constraints:
\begin{align}
    \sum_{u,v}\Tr_{A_IA_OB_IB_O}[W_{u,v}] &= d_Ad_B,  \nonumber\\
    \sum_{u,v}{}_{B_IB_O[1-A_I][1-A_O]}W_{u,v}&=0, \nonumber \\
    \sum_{u,v}{}_{A_IA_O[1-B_I][1-B_O]}W_{u,v}&=0,  \nonumber\\
    \sum_{u,v}{}_{[1-A_I][1-A_O][1-B_I][1-B_O]}W_{u,v}&=0. \nonumber
\end{align}
\end{theorem}
\begin{proof}
In this proof, we follow the technique of Araújo et al. in Appendix B of Ref.~\cite{araujo2015witnessing} to derive the necessary and sufficient conditions for a valid process matrix. 

First, we prove that the conditions stated in the theorem are necessary. Take two Hermitian operators
\begin{equation}
    x\in \mathcal{L}(\mathcal{H}^{A_I}\otimes\mathcal{H}^{A_O}),\quad  y\in \mathcal{L}(\mathcal{H}^{B_I}\otimes\mathcal{H}^{B_O})
\end{equation}
acting on the Hilbert spaces of Alice and Bob, respectively. Then we can construct operators
\begin{align}
    \mathcal{X}^{A_IA_O} = {}_{[1-A_I][1-A_O]}x + \mathbb{1}^{A_IA_O}/d_A, \quad \mathcal{Y}^{B_IB_O} = {}_{[1-B_I][1-B_O]}y + \mathbb{1}^{B_IB_O}/d_B
\end{align}
that satisfy the normalization constraints Eqn.~(\ref{dubcaus}) for any choice of $x$ and $y$. Requiring normalized probabilities when Alice and Bob use $\mathcal{X}$ and $\mathcal{Y}$ as their quantum operations amounts to requiring
\begin{align}
    \sum_{u,v}\Tr_{A_IA_OB_IB_O}\bigg[W_{u,v}^{A_IA_OB_IB_O}\vdot &\bigg({}_{[1-A_I][1-A_O]}x + \mathbb{1}^{A_IA_O}/d_A\bigg) \nonumber \\
    &\qquad\otimes \bigg({}_{[1-B_I][1-B_O]}y + \mathbb{1}^{B_IB_O}/d_B\bigg) \bigg] = 1
\end{align}
for all Hermitian $x$ and $y$. Then we can consider four possibilities:

\vspace{0.1in}
\noindent
($x=y=0$) In this case, the constraint reduces to 
\begin{equation}
    \sum_{u,v}\Tr_{A_IA_OB_IB_O}[W_{u,v}] = d_Ad_B.
\end{equation}

\vspace{0.1in}
\noindent
($x\neq 0, y=0$) The term from the $x=y=0$ remains and equals one, while the additional term for $x\neq 0$ must equal zero to preserve the normalization. The normalization constraint then reduces to
\begin{equation}
    \sum_{u,v}\Tr_{A_IA_OB_IB_O}[{}_{[1-A_I][1-A_O]}x\cdot W_{u,v}^{A_IA_OB_IB_O}]=0.
\end{equation}
Now we make use of the following fact. The pre-subscript ${}_{X}M$ acting on a Hermitian operator $M\in \mathcal{L}(\mathcal{H}_X\otimes\mathcal{H}_Y)$ is self-adjoint in the Hilbert-Schmidt inner product, that is,
\begin{equation}
    \Tr_{XY}[{}_{X}M_1^{XY}\cdot M_2^{XY}] = \Tr_{XY}[M_1^{XY}\cdot{}_{X}{M_2}^{XY}].
\end{equation}
This can be shown simply,
\begin{align}
    \Tr_{XY}[(\mathbb{1}_X\otimes\Tr_X[M_1^{XY}])\cdot M_2^{XY}] &= \Tr_Y[\Tr_X[M_1^{XY}]\cdot\Tr_X[M_2^{XY}]] \nonumber \\
    &= \Tr_{XY}[M_1^{XY}\cdot (\mathbb{1}_X\otimes\Tr_X[M_2^{XY}])].
\end{align}
Following the same logic, the operation ${}_{[1-X]}M$ acting on $M$ is self-adjoint in the Hilbert-Schmidt inner product. This requires that
\begin{equation}
    \sum_{u,v}\Tr_{A_IA_OB_IB_O}[x\cdot {}_{[1-A_I][1-A_O]}W_{u,v}^{A_IA_OB_IB_O}]=0
\end{equation}
for all Hermitian $x$. The only Hermitian operator that is orthogonal to all other Hermitian operators is the zero operator. Therefore,
\begin{equation}
    \sum_{u,v}\Tr_{B_IB_O}[{}_{[1-A_I][1-A_O]}W_{u,v}]=0.
\end{equation}
This is equivalent to the second constraint in the theorem.

\vspace{0.1in}
\noindent
($x=0, y\neq 0$) Following the same steps but swapping Alice and Bob's roles gives the third constraint:
\begin{equation}
    \sum_{u,v}\Tr_{A_IA_O}[{}_{[1-B_I][1-B_O]}W_{u,v}]=0.
\end{equation}

\vspace{0.1in}
\noindent
($x,y \neq 0$) In this case the normalization constraint reduces to
\begin{equation}
    \sum_{u,v}\Tr_{A_IA_OB_IB_O}[({}_{[1-A_I][1-A_O]}x\otimes{}_{[1-B_I][1-B_O]}y)\cdot W_{u,v}^{A_IA_OB_IB_O}]=0.
\end{equation}
We can use the fact that the action of ${}_{[1-X]}\vdot$ is self-adjoint to have each of these act on $W$
\begin{equation}
    \sum_{u,v}\Tr_{A_IA_OB_IB_O}[(x\otimes y)\cdot {}_{[1-A_I][1-A_O]}{}_{[1-B_I][1-B_O]}W_{u,v}^{A_IA_OB_IB_O}]=0.
\end{equation}
Since this is true for all Hermitian $x$ and $y$, we find that
\begin{equation}
    \sum_{u,v}{}_{[1-A_I][1-A_O]}{}_{[1-B_I][1-B_O]}W_{u,v}^{A_IA_OB_IB_O}=0.
\end{equation}
In conclusion, we have shown that the four constraints in the theorem are a necessary consequence of the four cases considered here.

Finally, we show that the four conditions in the theorem are sufficient for normalized probabilities. We rewrite Alice's averaged operation in the form
\begin{equation}
    M^{A_IA_O} = {}_{[1-A_I][1-A_O]}M^{A_IA_O} + {}_{A_I}M^{A_IA_O} + {}_{A_O}M^{A_IA_O} - {}_{A_IA_O}M^{A_IA_O}
\end{equation}
by trivially adding and subtracting the partial traces of $M^{A_IA_O}$. Alice's operation is required to satisfy normalization constraints as in Eqn.~(\ref{dubcaus}). We can use this to write
\begin{equation}
    M^{A_IA_O} = {}_{[1-A_I][1-A_O]}M^{A_IA_O} + \mathbb{1}^{A_IA_O}/d_A.
\end{equation}
We can write Bob's operation in the analogous way. Inserting these expressions into Eqn.~(\ref{sumtoone}) gives
\begin{align}                       
  \sum_{u,v}\Tr_{A_IA_OB_IB_O}\bigg[W_{u,v}^{A_IA_OB_IB_O} \vdot &\bigg({}_{[1-A_I][1-A_O]}M^{A_IA_O} + \mathbb{1}^{A_IA_O}/d_A\bigg) \nonumber \\
  &\quad \otimes \bigg({}_{[1-B_I][1-B_O]}M^{B_IB_O} + \mathbb{1}^{B_IB_O}/d_B\bigg)\bigg] = 1.
\end{align}
Expanding the tensor product results in four terms. The term
\begin{equation}
    \sum_{u,v}\Tr_{A_IA_OB_IB_O}[\mathbb{1}^{A_IA_O}/d_A\otimes \mathbb{1}^{B_IB_O}/d_B\cdot W_{u,v}^{A_IA_OB_IB_O}]
\end{equation}
is equal to one if the first constraint from the theorem is met by the process matrix $W_{u,v}$. The remaining three terms vanish due to the remaining three constraints in the theorem. Therefore, the four conditions are sufficient to guarantee normalized probabilities.

\end{proof}

\section{Process matrices without pre-selection and post-selection}
\label{nopostpre}

One of the main lessons from our study of processes with definite causal order in Section~\ref{causineq} is that signalling forward in time is not possible without pre-selection, and signalling backward in time is not possible without post-selection. This guiding principle imposes further constraints on process matrices that do not feature pre-selection or post-selection. 

\begin{figure}[htbp] 
        \makebox[\textwidth][c]{        \includegraphics[width=1.2\textwidth]{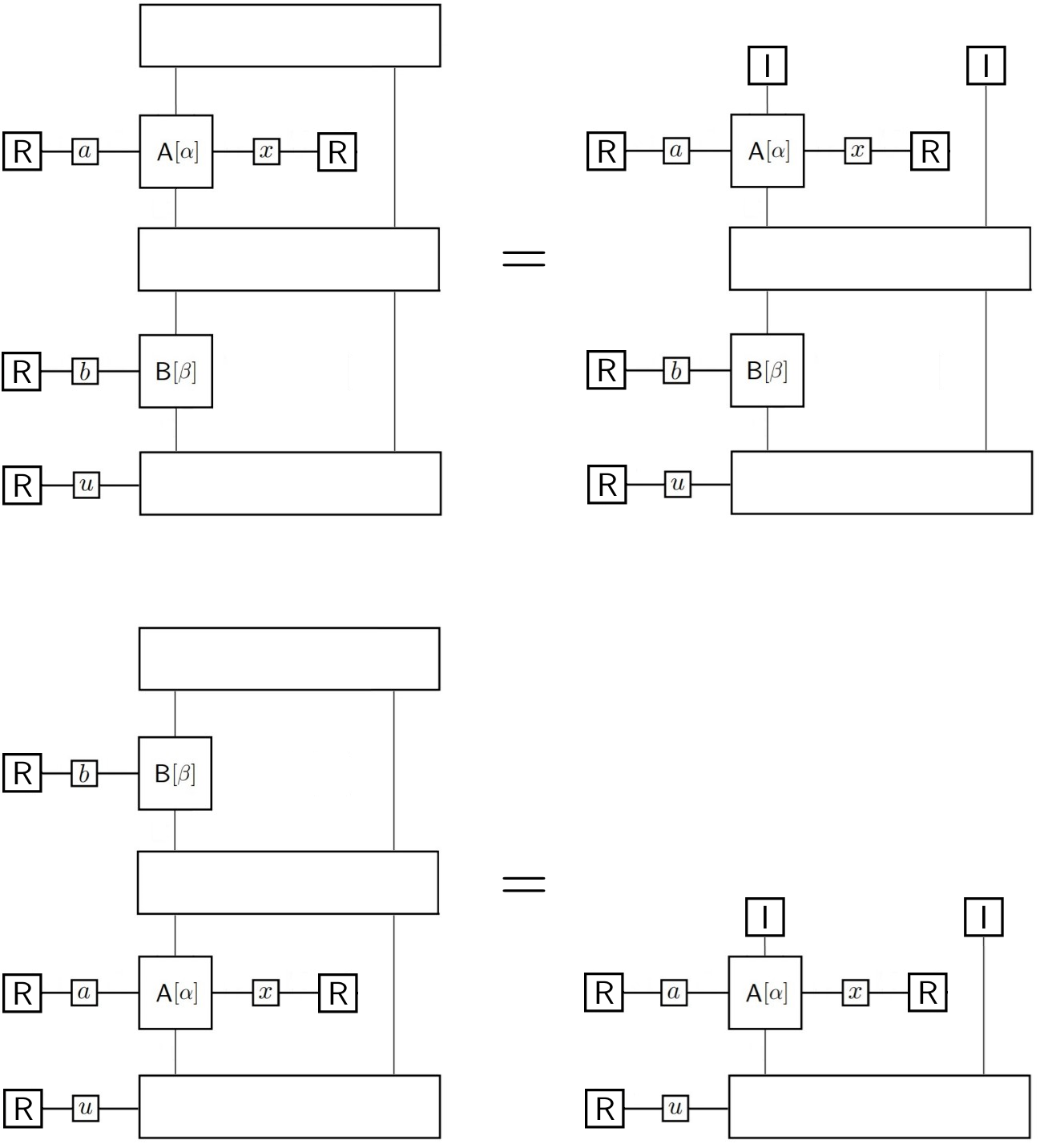}}
        \caption{
                \label{nops} 
                {With $v$ and $y$ marginalized, Bob cannot signal backward in time to Alice. In either definite causal order, $A\preceq B$ or $B\preceq A$, applying the double causality rules shows that Alice's output must be ignored.}           
        }
\end{figure}
Post-selection can be accomplished either by conditioning on a value for the post-selection variable $v$ in the process matrix, or by conditioning on Alice or Bob's outcome $x$ or $y$. See Fig.~\ref{abprocess} for the definitions of these variables. Consider first the case when $v$ and $y$ are marginalized. For causally separable processes, it is clear that Bob cannot signal backward in time to Alice. In diagrammatic notation, this fact is demonstrated in Fig.~\ref{nops}. By applying the double causality rules, it can be seen that Alice's output is always ignored in either definite causal order, $A\preceq B$ or $B\preceq A$. We assert that even in a causally non-separable process, Bob cannot signal backward in time to Alice with this marginalization, so Alice's outcome must be ignored. In quantum theory, this can be written
\begin{align}
    \sum_{v,y}&\Tr_{A_IA_OB_IB_O}[W^{A_IA_OB_IB_O}_{u,v}\vdot(M_{a,x}^{A_IA_O}\otimes M_{b,y}^{B_IB_O})] \nonumber \\
    &= \sum_{v,y}\Tr_{A_IA_OB_IB_O}[{}_{A_O}W^{A_IA_OB_IB_O}_{u,v}\vdot(M_{a,x}^{A_IA_O}\otimes M_{b,y}^{B_IB_O})].
\end{align}
This must hold for all operations $M_{a,x}^{A_IA_O}$ and $M_{b,y}^{B_IB_O}$. Therefore we find that 
\begin{equation}
    \Tr_{A_IA_OB_IB_O}\bigg[\sum_{v}{}_{[1-A_O]}W^{A_IA_OB_IB_O}_{u,v}\vdot\bigg(M_{a,x}^{A_IA_O}\otimes \sum_y M_{b,y}^{B_IB_O}\bigg)\bigg] = 0.
\end{equation}
This is a statement of the orthogonality between $\sum_{v}{}_{[1-A_O]}W^{A_IA_OB_IB_O}_{u,v}$ and the subspace generated by operators of the form $M_{a,x}^{A_IA_O}\otimes \sum_y M_{b,y}^{B_IB_O}$ under the Hilbert-Schmidt inner product. Keeping in mind that Alice and Bob's operations satisfy the double causality rules Eqn.~(\ref{dubcaus}), this results in the constraints Eqns.~(\ref{vcons2}, \ref{vcons3}). Marginalizing over $v$ and $x$ instead and repeating this analysis results in the remaining constraint Eqn.~(\ref{vcons1}).

A similar argument can be made for processes without pre-selection. Consider the case when $u$ and $a$ are marginalized. Then Alice cannot signal forward in time to Bob, and Bob's input must be ignored. This translates into the orthogonality equation:
\begin{equation}
    \Tr_{A_IA_OB_IB_O}\bigg[\sum_{u}{}_{[1-B_I]}W^{A_IA_OB_IB_O}_{u,v}\vdot\bigg(\sum_aM_{a,x}^{A_IA_O}\otimes M_{b,y}^{B_IB_O}\bigg)\bigg] = 0
\end{equation}
for all operations $M_{a,x}^{A_IA_O}$ and $M_{b,y}^{B_IB_O}$. Then a process matrix without pre-selection must satisfy Eqns.~(\ref{ucons2}, \ref{ucons3}). Repeating the argument by marginalizing $u$ and $b$ results in the final constraint Eqn.~(\ref{ucons1}).

\bibliographystyle{unsrt}
\bibliography{output}

\end{document}